\documentclass[reprint,pra, twocolumn, amsmath,amssymb,
 aps,nofootinbib,superscriptaddress,floatfix]{revtex4-2}

\usepackage[utf8]{inputenc}
\usepackage{xcolor}
\usepackage{mathrsfs}  
\usepackage{lineno}
\usepackage{amsfonts}
\usepackage{amsthm}
\usepackage{dsfont}
\usepackage{physics}
\usepackage{mathpazo}
\usepackage{comment}
\usepackage{mathtools}
\usepackage{courier}
\usepackage{lmodern}

\usepackage{anyfontsize}
\usepackage{graphics}
\usepackage{bm}
\usepackage{bbm}
\usepackage{amsmath}
\usepackage{amssymb}
\usepackage{graphicx}
\usepackage{bbold}
\usepackage[titletoc]{appendix}
\usepackage{courier}
\usepackage{hyperref}
\usepackage{xpatch}
\usepackage{soul}
\usepackage[normalem]{ulem}
\definecolor{airforceblue}{rgb}{0.36, 0.54, 0.66}
\definecolor{darkolivegreen}{rgb}{0.33, 0.42, 0.18}
\definecolor{cerulean}{rgb}{0.0, 0.48, 0.65}

\hypersetup{
    colorlinks=true,
    linkcolor=cerulean,
    filecolor=magenta,
    citecolor=darkolivegreen,
    urlcolor=airforceblue,
    pdftitle={Robust self-testing of Bell inequalities tilted for maximal loophole-free nonlocality},
    pdfpagemode=FullScreen,
}

\usepackage{float} 
\usepackage{ulem}

\newcommand{\mbraket}[1]{\langle #1 \rangle}
\newcommand{\braketm}[3]{\langle #1 \hspace{1pt} |\hspace{1pt} #2 \hspace{1pt}| \hspace{1pt} #3 \rangle}
\newtheorem{theorem}{Theorem}

\newtheorem{lemma}{Lemma}

\newcommand{\FC}{\mathcal{F}}

\newcommand{\LC}{\mathcal{L}}

\newcommand{\QC}{\mathcal{Q}}

\usepackage{listings} 

\begin{document}
\title{Self-testing tilted strategies for maximal loophole-free nonlocality}

\author{Nicolas Gigena}\affiliation{IFLP/CONICET--Departamento  de F\'{\i}sica, Universidad Nacional de La Plata, C.C. 67, La Plata (1900), Argentina}

\author{Ekta Panwar}

\affiliation{Institute of Theoretical Physics and Astrophysics, Faculty of Mathematics, Physics and Informatics, University of Gdańsk, ul. Wita Stwosza 57, 80-308 Gdańsk, Poland}
\affiliation{International Centre for Theory of Quantum Technologies, University of Gdańsk, 80-308 Gdańsk, Poland}
\author{Giovanni Scala}
\affiliation{Dipartimento Interateneo di Fisica, Politecnico di Bari, 70126 Bari, Italy}
\affiliation{International Centre for Theory of Quantum Technologies, University of Gdańsk, 80-308 Gdańsk, Poland}
\affiliation{INFN, Sezione di Bari, 70126 Bari, Italy}
\author{Mateus Araújo}\affiliation{Departamento de Física Teórica, Atómica y Óptica, Universidad de Valladolid, 47011 Valladolid, Spain}

\author{Máté Farkas}
\affiliation{Department of Mathematics, University of York, Heslington, York, YO10 5DD, United Kingdom}

\author{Anubhav Chaturvedi}\email{anubhav.chaturvedi@pg.edu.pl}
\affiliation{Faculty of Applied Physics and Mathematics,
 Gdańsk University of Technology, Gabriela Narutowicza 11/12, 80-233 Gdańsk, Poland}
 \affiliation{International Centre for Theory of Quantum Technologies, University of Gdańsk, 80-308 Gdańsk, Poland}
\begin{abstract}
The degree of experimentally attainable nonlocality, as gauged by the loophole-free or effective violation of Bell inequalities, remains severely limited due to inefficient detectors. We address an experimentally motivated question: Which quantum strategies attain the maximal loophole-free nonlocality in the presence of inefficient detectors? For \emph{any} Bell inequality and \emph{any} specification of detection efficiencies, the optimal strategies are those that maximally violate a \emph{tilted} version of the Bell inequality in ideal conditions. In the simplest scenario, we demonstrate that the quantum strategies that maximally violate the \emph{doubly-tilted} versions of \emph{Clauser-Horne-Shimony-Holt} inequality are \emph{unique} up to local isometries. We utilize a Jordan's lemma and Gr\"obner basis-based proof technique to analytically derive self-testing statements for the \emph{entire} family of doubly-tilted CHSH inequalities and numerically demonstrate their robustness. These results enable us to reveal the insufficiency of even high levels of the \emph{Navascu\'es--Pironio--Ac\'in} hierarchy to saturate the maximum quantum violation of these inequalities.
\end{abstract}


\maketitle
 
\section*{Introduction}
Correlations born of local measurements performed on entangled quantum systems shared between distant observers resist local-causal explanations, a phenomenon known as \emph{Bell nonlocality} \cite{Bell1964, Brunner2014}. Apart from their foundational significance, nonlocal correlations enable several classically impossible information processing and cryptographic applications such as unconditionally secure Device-Independent Quantum Key Distribution (DIQKD) \cite{Ekert1991, mayers1998quantum, BHK05, Acin2007, PABGMS09,Ghoreishi2025}. The efficacy of these applications relies on loophole-free certification of strong nonlocal correlations. In particular, the \emph{detection loophole}, posed by the lack of perfect detectors in Bell tests, is the most persistent obstacle in the experimental realization of strong long-range loophole-free nonlocal correlations. 
The detection efficiency of a measuring party is the ratio of particles detected to the total number of particles emitted by the source. The effective detection efficiency $\eta$ depends on the detectors and decays with the distance from the source. Closing the detection loophole in Bell experiments amounts to having an effective detection efficiency $\eta$ higher than a threshold value $\eta^*$, referred to as the \emph{critical detection efficiency}.

Consequently, significant research efforts have been directed towards minimizing the critical detection efficiency requirement for loophole-free certification of nonlocality \cite{Eberhard1993,Massar2002, Vertesi2010,miklin2022exponentially,chaturvedi2024extending}. However, for real-world applications to be effective, mere violation of a Bell inequality is insufficient \cite{farkas2021}. Instead, the efficacy of such applications \cite{chaturvedi2024extending} typically requires a high degree of nonlocality and motivates the question: 
\begin{quote}
    Which quantum strategies yield the maximum loophole-free nonlocality in the presence of inefficient detectors? 
\end{quote}

As the extent of the violation of a (facet) Bell inequality corresponds to the distance of a nonlocal correlation from a facet of the local polytope, it translates to a reliable measure of nonlocality \cite{araujo2020bell}. Thus the question above boils down to finding the quantum strategies that yield the maximal loophole-free violation of a Bell inequality for specified detection efficiencies. Since the use of inefficient detectors results in the occurrence of ``no-click" events, to decide whether a Bell inequality is violated in a loophole-free way, these events must be included in the measurement statistics. An experimentally convenient way of including the ``no-click" events is to assign a valid outcome to them \cite{Eberhard1995}. Moreover, such local \emph{assignment} strategies have been proven to be optimal in the simplest Bell scenario \cite{PhysRevA.83.032123}. We use local assignment strategies to show that the quantum state and measurements maximally violating a given Bell inequality, in the presence of inefficient detectors, correspond to those maximally violating, under ideal conditions, a tilted version of the inequality. In the simplest Bell scenario, where up to a relabeling of measurements and outcomes the only facet inequality is the \emph{Clauser-Horne-Shimony-Holt} (CHSH) inequality \cite{Clauser1969}, attaining maximal loophole-free nonlocality amounts to maximally violating a ``doubly-tilted" CHSH inequality.

The doubly-tilted CHSH inequalities are generalizations of the tilted CHSH inequalities considered in \cite{tiltedasym, acin2012versus}, which correspond to the special case of one party having access to ideal detectors, and for which both the maximal violation and the quantum realization attaining it are known. In fact, in \cite{tiltedasym} it is shown using \emph{sum-of-squares} (SOS) decompositions that the optimal state and measurements are unique up to local unitaries, i.e., the maximal violation \emph{self-tests} the optimal quantum strategy. In contrast, we analytically derive self-testing statements for the entire family doubly-tilted CHSH inequalities as a function of the detection efficiencies via a novel proof technique based on Gröbner basis elimination (for an introduction to the method for a physics audience see Appendix A of \cite{Lee2017}) and Jordan's lemma (Chapter VII \cite{bhatia2013matrix}). We find that the optimal quantum strategy entails non-maximally incompatible observables for both parties and a partially entangled two-qubit state, with the optimal observables of a party also depending on the detection efficiency of the other party.

The analytical self-testing statements enable us to recover, up to arbitrary precision, the maximum quantum violation of the doubly-tilted CHSH inequalities. We compare these values with the corresponding upper bounds from the Navascués-Pironio-Acín (NPA) hierarchy \cite{Navascues2007, Navas_2008}. Strikingly, even higher levels (up to level 10) cannot saturate the maximal violation of the doubly-tilted CHSH inequalities. Consequently, the SOS decomposition-based self-testing technique \cite{tiltedasym} turns out to be analytically intractable for the doubly-tilted CHSH inequalities. We conclude by discussing implications of our self-testing statements towards device-independent cryptography with inefficient detectors and the complexity of characterizing the set of quantum correlations via the NPA hierarchy and the SOS decompositions self-testing method.

\section*{Results}
\subsection{Nonlocality with imperfect detectors} \label{sec_2}

Consider a bipartite Bell experiment wherein a source distributes a physical system to be shared between Alice and Bob. Alice and Bob perform one out of $m_A$, $m_B$ possible measurements, labeled $x$, $y$, which have $d_A$, $d_B$ possible outcomes labeled $a$, $b$, respectively. The four-tuple $(m_A,m_B,d_A,d_B)$ specifies the Bell scenario. The experiment results in the behavior vector $\bm{p}\in \mathbb{R}^{m_A d_A m_B d_B}_+$ with entries $p(ab|xy)$ specifying the probability of outcomes $a$, $b$ when Alice measures $x$ and Bob measures $y$. We denote by ${\bm p}^A\in\mathbb{R}^{m_A d_A}_+$, ${\bm p}^B \in\mathbb{R}^{m_B d_B}_+$ the marginal vectors with entries $p_A(a|x),p_B(b|y)$ specifying the respective marginal probabilities. The set $\LC$ of behaviors $\bm p$ admitting a local-causal explanation forms a polytope, with local deterministic strategies $\bar{\bm{p}}$ as vertices satisfying $\bar{p}(ab|xy)=\delta_{a,a_x}\delta_{b,b_y}$, where $a_x$, $b_y$ are local deterministic \emph{assignments}. 

A quantum strategy is described by the three-tuple $(\hat{\rho}_{AB},\{ \hat{M}_a^x\}, \{\hat{N}_b^y\})$ entailing the shared quantum state $\hat{\rho}_{AB}$ and local quantum measurement operators $\{ \hat{M}_a^x\}, \{\hat{N}_b^y\}$ which results in a quantum behavior $\bm p$ with components,
\begin{equation}
    p(ab|xy) = \Tr \; \left(\hat{\rho}_{AB}\, \hat{M}_a^x\otimes \hat{N}_b^y\right) \qquad \forall\; a, b, x, y \;, \label{quantum_prob}
\end{equation}
We denote the convex set of quantum behaviors by $\QC$. It is well know that $\LC\subseteq \QC$ and that there exists \emph{nonlocal} quantum behaviors such that $\bm{p}\in\QC\setminus \LC$. The hyperplanes containing the facets of $\LC$ thus separate the local behaviors from those that are \emph{nonlocal}, and each of them is associated with a Bell inequality of the form,
\begin{equation}
    \beta(\bm p) := \bm\beta\cdot\bm p =\sum_{abxy}\beta_{abxy}p(ab|xy) \leq \beta_\LC, \label{general_bell_inequality}
\end{equation}

where $(\cdot)$ denotes the inner product in $\mathbb{R}^{m_A d_A m_B d_B}_+$, and the vector $\bm\beta\in \mathbb{R}^{m_A d_A m_B d_B}_+$ with the real coefficients $\beta_{abxy}$ as entries specifies the \emph{Bell functional} $\beta(\bm{p})$ and lies in the orthogonal complement of the corresponding facet. A nonlocal behavior $\bm{p} \notin \LC$ violates at least one such facet Bell inequality, such that, $\beta(\bm{p})>\beta_{\LC}$. The amount of the violation, $\beta(\bm{p})-\beta_\LC$, is related to the distance of $\bm p$ from $\LC$, and constitutes a measure of nonlocality \cite{Araujo2020bellnonlocality}. We denote the maximal quantum value of the Bell functional $\beta(\bm{p})$ by $\beta_\QC$.

In the derivation of inequality \eqref{general_bell_inequality} the detectors are assumed to be perfect. However in actual experiments detectors sometimes fail to detect an incoming system, which results in the occurrence of ``no-click" events. The most general way of accounting for a ``no-click" event is to consider it an additional outcome of the local measurements, which enlarges the considered Bell scenario. A more convenient approach, which preserves the Bell scenario, consists of locally assigning a pre-existing outcome to each ``no-click'' event \cite{Eberhard1995,miklin2022exponentially,XuLatest}. Moreover, this method is particularly well-suited for our purposes as we are interested in gauging the effect of imperfect detectors on the violation of a given facet Bell inequality.

Let $\eta_A,\eta_B \in [0,1]$ be Alice's and Bob's independent detector efficiencies. A local assignment strategy is described by a local behavior $\bm{q}\in \LC$ with $q(ab|xy)$ as entries specifying the probability with which the parties assign the outcomes $a$, $b$ to failed measurements of $x$, $y$, respectively. Then the effect of inefficient detectors on an ideal behavior $\bm{p}$ can be summarized as the following affine map,
\begin{align} \label{effectiveP}
    \tilde{\bm p} = &\eta_A\eta_B{\bm p} + \eta_A(1-\eta_B){\bm p}^A\otimes {\bm q}^B   \nonumber\\
    &+ (1-\eta_A)\eta_B {\bm q}^A \otimes {\bm p}^B + (1-\eta_A)(1-\eta_B)\bm q,
\end{align}
where $\otimes$ is the Kronecker product, $\bm{q}^A \in\mathbb{R}^{m_A d_A}_+$, $\bm{q}^B  \in\mathbb{R}^{m_B d_B}_+$ are the respective marginals of $\bm{q}$. 

Observe that in \eqref{effectiveP}, the effective behavior $\tilde{\bm{p}}$ is a convex mixture of the ideal behavior $\bm{p}$, the product behaviors  ${\bm p}^A\otimes {\bm q}^B$, ${\bm q}^A \otimes {\bm p}^B$ and the local assignment strategy $\bm{q}$. Since, ${\bm p}^A\otimes {\bm q}^B,{\bm q}^A \otimes {\bm p}^B,\bm{q} \in \LC$, locally assigning outcomes to failed measurements cannot increase the value of the Bell functional $\beta(\tilde{\bm p})$ \eqref{general_bell_inequality} beyond the local bound $\beta_\LC$. Hence, we say that a given Bell inequality \eqref{general_bell_inequality} is violated in a \emph{loophole-free} way if the effective behavior $\tilde{\bm p}$ violates it, such that $\beta(\tilde{\bm p})>\beta_\LC$. For any given $\eta_A,\eta_B$, we are interested in the maximal effective or loophole-free violation of a Bell inequality \eqref{general_bell_inequality}, 
\begin{equation} \label{lfviolation}
\beta(\tilde{\bm p})-\beta_\LC. 
\end{equation}
As $\tilde{\bm p}$ depends on the ideal quantum behavior $\bm{p}\in \QC$ and the local assignment strategy $\bm{q}\in\LC$ \eqref{effectiveP}, finding the maximal loophole-free violation \eqref{lfviolation} of a Bell inequality, requires optimizing over both $\QC$ and $\LC$. We now present a useful Lemma which resolves this optimization,
\begin{lemma}\label{tiltedTheorem}
For any specification detection efficiencies $\eta_A,\eta_B$ and any given Bell inequality $\beta(\bm p) \leq \beta_\LC,$ \eqref{general_bell_inequality}, the ideal quantum behaviors $\bm{p}\in \QC$ that yield the maximal loophole-free violation $\beta(\tilde{\bm p})-\beta_\LC$ \eqref{lfviolation}, are the ones that maximally violate a tilted Bell inequality $\bm{\beta}_{\eta_A,\eta_B}(\bm{p})\leq \beta_{\LC}(\eta_A,\eta_B)$ where $\beta_\LC(\eta_A,\eta_B)= \beta_{\LC}/{\eta_A\eta_B}-\beta(\bar{\bm q})\alpha(\eta_A)\alpha(\eta_B)$ and,
\begin{align} \label{tilted_B}
\bm{\beta}_{\eta_A,\eta_B}(\bm{p})= \bm{\beta}\cdot\bm{p} + \alpha(\eta_B)\bm{\beta}^A_{\bar{\bm q}}\cdot\bm{p}^A + \alpha(\eta_A)\bm{\beta}^B_{\bar{\bm q}}\cdot\bm{p}^B,
\end{align}
where $\alpha(\eta)=(1-\eta)/\eta$, $\bar{\bm q}$ is a deterministic assignment strategy with entries $\bar{q}(ab|xy)=\delta_{a,a_x}\delta_{b,b_y}$, and $\bm{\beta}^A_{\bar{\bm q}}$, $ \bm{\beta}^B_{\bar{\bm q}}$ represent single party Bell functionals with coefficients $\beta^A_{ax}=\sum_y\beta_{ab_yxy}$, $\beta^B_{by}=\sum_x\beta_{a_xbxy}$ as entries, respectively, such that, the effective value of the original Bell functional is $\beta(\tilde{p})=\eta_A\eta_B\left(\beta_Q(\eta_A,\eta_B)-\alpha(\eta_A)\alpha(\eta_B) \bm{\beta}\cdot \bar{\bm{q}}\right)$ where $\beta_Q(\eta_A,\eta_B)$ is the maximum quantum value of the Bell functional $\beta_{\eta_A,\eta_B}$.
\end{lemma}

The proof has been deferred to the Supplementary Information for brevity. As a consequence of Lemma \ref{tiltedTheorem}, finding the maximal loophole-free violation \eqref{lfviolation} of a Bell inequality amounts to finding the maximum quantum values of a finite set of tilted Bell functionals of the form \eqref{tilted_B} associated with the finite set of deterministic vertices of $\LC$. In the next section, we exemplify the application of Lemma \ref{tiltedTheorem} and find the maximum loophole-free nonlocality in the simplest Bell scenario for any specification of detection efficiencies.

\subsection{Maximal loophole-free nonlocality in the CHSH scenario} \label{sec_3}

We now consider the simplest bipartite Bell scenario, in which Alice and Bob perform one of the two distinct measurements, $x, y \in\{0, 1\}$ and obtain binary outcomes, $a, b\in\{-1, +1\}$, respectively. It is convenient in this scenario to introduce \textit{correlators}, $\mbraket{A_x B_y}:=\sum_{ab} ab\, p(ab|xy)$, and marginals $\mbraket{A_x}:=\sum_{a} a\, p_A(a|x)$, $\mbraket{B_y}:=\sum_{b} b\, p_B(b|y)$
\cite{Brunner2014}. In a quantum strategy these averages correspond to the expectation values of binary observables, $\hat{A}_x = \hat{M}^x_{+1} - \hat{M}^x_{-1}$ and $\hat{B}_y = \hat{N}^y_{+1} - \hat{N}^y_{-1}$ with respect to a shared quantum state $\hat{\rho}_{AB}$. Then the nonlocality of a given behavior $\bm{p}$ can be witnessed by a violation of the CHSH inequality, 
\begin{equation}
    C(\bm p) = \sum_{x, y}(-1)^{x\cdot y}\mbraket{A_x B_y}\leq 2, \label{CHSH_ineq}
\end{equation}
which up to relabeling of measurements and outcomes is known to be the only tight and complete Bell inequality in this scenario \cite{Gigena2022,Brunner2014review}. Hence, the violation of the CHSH inequality $C(\bm p)-2$ forms our measure of nonlocality. 

We are interested in the maximal loophole-free violation of the CHSH inequality $C(\tilde{\bm p})-2$ for any specification of detection efficiencies $\eta_A,\eta_B$. Therefore, we invoke Lemma \ref{tiltedTheorem} with a  
deterministic assignment strategy, wherein $q(ab|xy)=\delta_{a,+1}\delta_{b,+1}$ for all $x,y$, to retrieve the following doubly-tilted CHSH inequality,

\begin{align}
    C_{\eta_A\eta_B}(\bm p) &= C({\bm p}) + \frac{2}{\eta_B}(1-\eta_B) \mbraket{A_0} + \frac{2}{\eta_A}(1-\eta_A) \mbraket{B_0}\nonumber\\
    &\leq  2 \left[ \frac{1}{\eta_A} + \frac{1}{\eta_B} -1 \right]= c_\LC(\eta_A,\eta_B). \label{tilted_CHSH}
\end{align}
 In the Supplementary Information we demonstrate that the assignment strategy considered above is optimal for all $\eta_A,\eta_B$. Consequently, for any given $\eta_A,\eta_B$, the quantum strategies which maximally violate the doubly-tilted CHSH inequality \eqref{tilted_CHSH} yield the maximal loophole-free violation of the CHSH inequality \eqref{CHSH_ineq}. However, not all combinations $\eta_A,\eta_B$ allow for such a violation. In particular, a quantum loophole-free violation of the CHSH inequality \eqref{CHSH_ineq} is not possible if the detection efficiencies $\eta_A,\eta_B\in[0,1]$ fail to satisfy \cite{Massar2003,Larsson2004,Cope2019,Lobo2024certifyinglongrange}, 
\begin{equation}
     \eta_B > \frac{\eta_A}{3\eta_A-1}.
     \label{critical_efficiency}
\end{equation}

Therefore, \eqref{critical_efficiency} provides lower bounds for Bob's critical detection efficiency $\eta^*_B\geq \frac{\eta_A}{3\eta_A-1}$ given Alice's detection efficiency $\eta_A\in (\frac{1}{2},1]$, effectively defining the region \eqref{critical_efficiency} of $\eta_A,\eta_B$, where we look for the maximum quantum violation of doubly-tilted CHSH inequalities, $c_{\QC}(\eta_A,\eta_B)$.  

In the following sections, we find the exact value of $c_{\QC}(\eta_A,\eta_B)$ along with the optimal quantum strategies. In FIG. \ref{fig:anufig2}, we plot the consequent maximal loophole-free value of the CHSH functional $C(\tilde{\bm p})=\eta_A\eta_B c_{\QC}(\eta_A,\eta_B)-2(1-\eta_A)(1-\eta_B)/\eta_A\eta_B$.

\begin{figure} 
    \centering
     \includegraphics[width=1\linewidth]{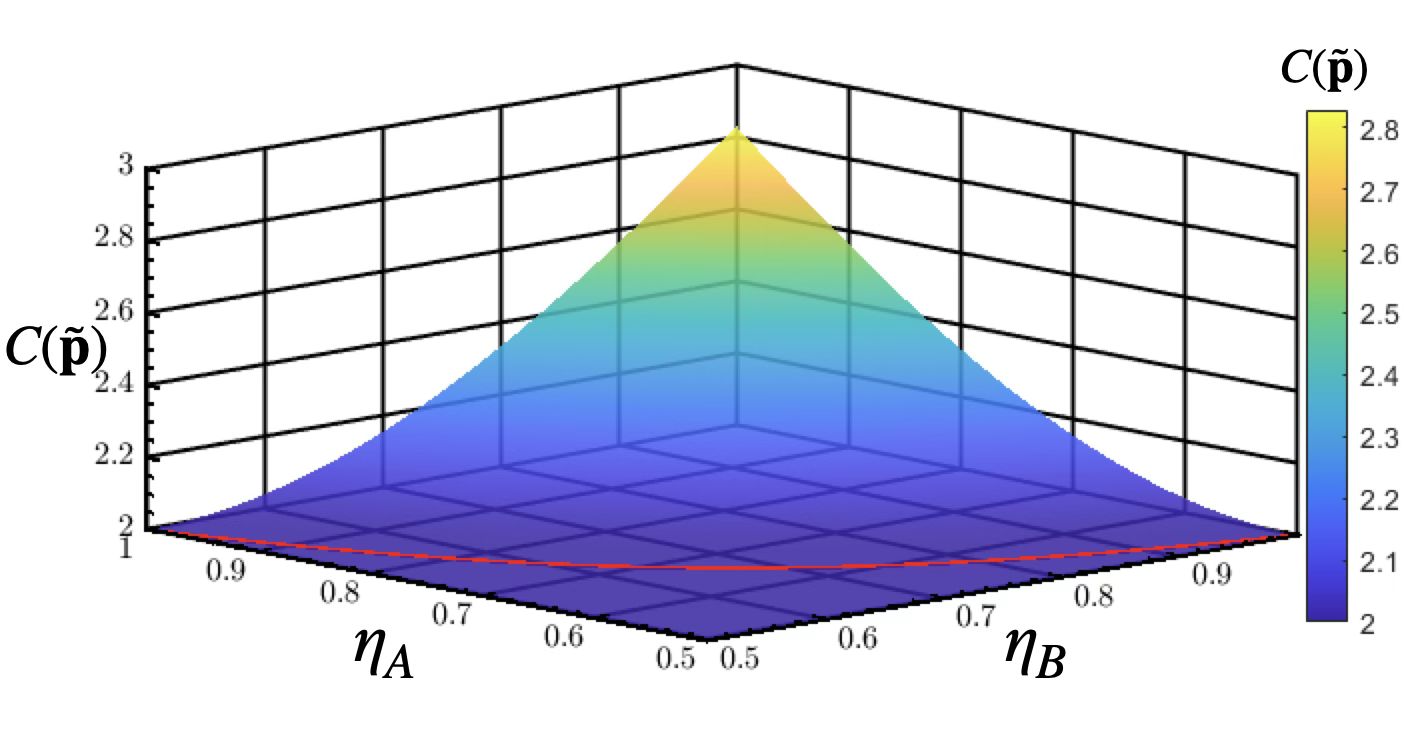}
    \caption{\emph{Maximum loophole-free violation of the CHSH inequality:---} A plot of the maximum loophole-free value of the CHSH functional, $C(\tilde{\bm{p}})=\eta_A\eta_Bc_{\QC}(\eta_A,\eta_B)+(1-\eta_A)(1-\eta_B)2$, against detection efficiencies $\eta_A,\eta_B\in[\frac{1}{2},1]$, where we used the analytical expression for maximum quantum violation of the doubly-tilted CHSH inequality \eqref{tilted_CHSH}, $c_{\QC}(\eta_A,\eta_B)$, derived in Section \ref{deriveAnalytic}.
    The solid red line represents Bob's critical detection efficiency $\eta^*_B=\frac{\eta_A}{3\eta_A-1}$ \eqref{critical_efficiency}, below which a loophole-free quantum violation of the CHSH inequality is not possible \cite{Massar2003,Larsson2004,Cope2019,Lobo2024certifyinglongrange}.
    }
    \label{fig:anufig2}
\end{figure}

\subsection{MAXIMAL VIOLATION OF DOUBLY-TILTED CHSH INEQUALITIES AND SELF-TESTING}

In this section we address the problem of computing the maximal violation of the inequalities of the form \eqref{tilted_CHSH}, which form a two-parameter family, and the quantum strategies attaining it. A one-parameter subset within this family is that corresponding to either $\eta_A=1$ or $\eta_B=1$, for which the solution is already known. We start by revisiting these results before presenting the solution for the general case which, remarkably, is far from being a straightforward generalization of the first. 
\subsection{One inefficient detector} \label{oneInefDetect}
Let us consider an \emph{ideal} scenario wherein Alice has access to perfect detectors such that $\eta_A=1$, while Bob's detectors are imperfect and click with efficiency $\eta_B$. Consequently, we retrieve the following family of tilted CHSH inequalities from \eqref{tilted_CHSH},
\begin{equation}
    C_{\alpha}(\bm p) = \sum_{x, y}(-1)^{x\cdot y}\mbraket{A_x B_y} + \alpha \mbraket{A_0} \leq  2+\alpha, \label{tiltedasym}
\end{equation}
where the tilting parameter $\alpha = \frac{2}{\eta_B}(1-\eta_B)$ is determined by the detection efficiency on Bob's side. First, a quantum loophole-free violation of the CHSH inequality \eqref{CHSH_ineq} requires $\eta_B\in(\frac{1}{2},1]$, which restricts the tilting parameter to $\alpha\in[0,2)$.  

The maximum quantum value of the Bell functional in ~\eqref{tiltedasym} is  $c_{\QC}(\alpha)= \sqrt{8 + 2\alpha^2 }$ \cite{acin2012versus}. Hence, the maximum loophole-free violation of CHSH inequality when $\eta_A=1$ is $2\sqrt{2}\sqrt{\eta^2_B+(1-\eta_B)^2}$. Recall that, in terms of the state and measurements $(\ket{\psi'}, \{\hat{A}'_x\}, \{\hat{B}'_y\})$ of an optimal quantum strategy, we can write the maximal quantum value of the tilted CHSH functional as $c_{\QC}(\alpha)=\braketm{\psi'}{\hat{C}_{\alpha}}{\psi'}$, where $\hat{C}_{\alpha}$ is the tilted CHSH Bell operator, given by
\begin{equation}
    \hat{C}_{\alpha} = \sum_{x, y}(-1)^{x\cdot y} \hat{A}_x\otimes \hat{B}_y + \alpha \hat{A}_0.
\end{equation}
In Ref.~\cite{tiltedasym} it is shown that $\hat{C}_{\alpha}$ admits an SOS decompositions of the form,
\begin{equation}
    c_{\QC}(\alpha)\hat{\mathbb{1}} - \hat{C}_{\alpha} = \sum_i \hat{P}_i^\dag \hat{P}_i, \label{SOS}   
\end{equation}
in terms of polynomials $\hat{P}_i$ in the operators $\{\hat{\mathbb{1}}, \{\hat{A}'_x\}, \{\hat{B}'_y\}, \{\hat{A}'_x\hat{B}'_y\}\}$. These decompositions are then used to prove that $C_{\alpha}(\bm p)=c_{\QC}(\alpha)$ self-tests the optimal strategy $(\ket{\psi}, \{\hat{A}_x\}, \{\hat{B}_y\})$ \cite{tiltedasym, Supic2020}, where,
\begin{equation} 
    \begin{split}
        \ket{\psi} &= \cos{\theta}\ket{00} + \sin{\theta}\ket{11}, \\ 
        \hat{A}_0 &= \sigma_z \qquad  \hat{B}_0 = \cos{\mu}\,\sigma_z + \sin{\mu}\,\sigma_x \\
        \hat{A}_1 &= \sigma_x \qquad  \hat{B}_1 = \cos{\mu}\,\sigma_z - \sin{\mu}\,\sigma_x, \label{measurements_asymmetric}
    \end{split} 
\end{equation}
with $\alpha = 2/\sqrt{1+2\tan^2{2\theta}}$, $\tan(\mu) = \sin(2\theta)$ and $\sigma_{x(z)}$ denotes the $x (z)$ Pauli matrix. 

The proof builds upon the observation that for \emph{any} optimal quantum strategy $(\ket{\psi'}, \{\hat{A}'_x\}, \{\hat{B}'_y\})$, the SOS decomposition \eqref{SOS} implies that $\ket{\psi'}$ must belong to the null space of the operators $\hat{P}_i$, i.e., it must satisfy the conditions, $\hat{P}_i\ket{\psi'}=0$ for all $i$. From these conditions it is possible to infer the existence of operators $\{\hat{Z}_A, \hat{X}_A, \hat{Z}_B, \hat{X}_B \}$ \cite{tiltedasym}, such that,
\begin{equation}
\begin{split}
    \hat{Z}_A\ket{\psi'} &= \hat{Z}_B\ket{\psi'} \\
    \sin\theta \hat{X}_A(\hat{\mathbb{1}}+\hat{Z}_B)\ket{\psi'} &= \cos\theta \hat{X}_A(\hat{\mathbb{1}}-\hat{Z}_A)\ket{\psi'}. \label{isometry_conditions}
\end{split}
\end{equation}
The conditions \eqref{isometry_conditions} ensure the existence of \emph{local isometries}, $\Phi_A$ and $\Phi_B$, mapping any optimal strategy $(\ket{\psi'}, \{\hat{A}'_x\}, \{\hat{B}'_y\})$ to the reference strategy $(\ket{\psi}, \{\hat{A}_x\}, \{\hat{B}_y\})$ in \eqref{measurements_asymmetric}, that is,
\begin{equation}
    \begin{split}
        \Phi_A\otimes\Phi_B (\ket{\psi'}) &= \ket{\psi}\otimes\ket{\text{junk}} \\
        \Phi_A\otimes\Phi_B (\hat{A}'_x\otimes \hat{B}'_y \ket{\psi'}) &= \hat{A}_x\otimes \hat{B}_y \ket{\psi}\otimes \ket{\text{junk}}, \label{isometries}
    \end{split}
\end{equation}
where $\ket{\text{junk}}$ represents the arbitrary state of additional degrees of freedom on the which the measurements act trivially. Thus the optimal quantum strategy is unique up to local isometries. We note that while the proof of relations in Eq.~\eqref{isometry_conditions} requires an ideal behavior, \cite{tiltedasym} demonstrates that self-testing can be made \emph{robust}. We recall that a Bell expression provides a robust self-test for a given quantum strategy, say $R$ if, in a noise-tolerant manner, the expected value \( \beta(\bm{p}) \) close to the maximum \( \beta_{\QC} \) \emph{consistently} corresponds to a strategy $\tilde{R}$ in a close neighborhood of the reference strategy $R$, up to local isometries.


\subsection{Two inefficient detectors} \label{deriveAnalytic}

We now consider the most generic experimental setting, wherein the detectors of both parties may be imperfect and click with efficiencies $\eta_A,\eta_B \in [0,1]$. We rewrite for convenience, the doubly-tilted CHSH inequalities \eqref{tilted_CHSH} for the general case as,
\begin{equation}
    C_{\alpha,\beta}(\bm p) = \sum_{x, y}(-1)^{x\cdot y}\mbraket{A_x B_y} + \alpha \mbraket{A_0} + \beta \mbraket{B_0} \leq  2+\alpha + \beta, \label{general_tilted}
\end{equation}
where the tilting parameter $\alpha = \frac{2}{\eta_B}(1-\eta_B)$ for Alice's term $\mbraket{A_0}$ depends on the Bob's detection efficiency $\eta_B$, while Bob's tilting parameter $\beta = \frac{2}{\eta_A}(1-\eta_A)$ is determined by Alice's efficiency $\eta_A$. Consequently, the boundary of the region of interest \eqref{critical_efficiency} translates to $0<\alpha + \beta < 2$ in terms of the tilting parameters. We are interested in finding the maximal quantum value $c_{\QC}(\alpha,\beta)$ of the doubly-tilted CHSH functional in \eqref{general_tilted}, as well as a quantum strategy attaining it, as a function of the tilting parameters $\alpha, \beta$. However, the problem is intractable via the analytical techniques from \cite{tiltedasym} described above. We discuss this feature in more detail in Section \ref{SOSintract}.


For the general doubly-tilted CHSH inequalities \eqref{general_tilted}, we consider here an alternative approach. Since in this scenario both Alice and Bob perform two dichotomic measurements, Jordan's lemma ensures the existence of local bases in which the two observables of each party take a jointly block diagonal form, with blocks of dimension at most two \cite{masanes2005extremal, supic2020selftestingof, Panwar2023}. It follows then that the ensuing Bell operator associated with the Bell functional in \eqref{general_tilted} also takes a block diagonal form. Hence, the quantum value can always be expressed as a convex combination of two-qubit values. Thus, the maximal quantum value of the Bell functional is always achievable with a two-qubit strategy. 



Consequently, without loss of generality we can take the local observables to be,
\begin{subequations}\label{qubitParam}
\begin{align} 
    \hat{A}_0&= {\sigma}_{Z} \ , \  \hat{A}_1=c_A\,{\sigma} _{Z}+s_A\,{\sigma}_{X} \\
     \hat{B}_0&= {\sigma}_{Z} \ , \  \hat{B}_1=c_B\,{\sigma}_{Z}+s_B\,{\sigma}_{X},
\end{align}
\end{subequations}
where $c_A = \cos\theta_A$, $s_A = \sin\theta_A$, $c_B = \cos\theta_B$ and $s_B = \sin\theta_B$, and $\theta_A,\theta_B\in(0,\frac{\pi}{2}]$ are the angles between Alice's and Bob's measurements, respectively. Finally, the doubly-tilted CHSH Bell operator can be expressed as,
\begin{equation}
    \hat{C}_{\alpha,\beta}=\sum_{x, y}(-1)^{x\cdot y}\hat{A}_x\otimes \hat{B}_y + \alpha\, \hat{A}_0\otimes \hat{\mathbb{1}}+ \beta\, \hat{\mathbb{1}}\otimes \hat{B}_0. \label{sym_tilted_functional}
\end{equation}
We are now prepared to present our main result, namely, the maximal quantum violation of the doubly-tilted CHSH inequalities \eqref{general_tilted} which self-tests the optimal quantum state and measurements attaining it. Here, we consider the symmetric case with $\alpha=\beta$.



\begin{theorem}\label{selfTestSymmetric}[\emph{Self-testing with symmetrically tilted CHSH inequalities}]
The maximal quantum value $c_{\mathcal{Q}}(\alpha,\alpha)$ of the symmetrically $(\alpha=\beta)$ tilted CHSH functional in \eqref{general_tilted} for $\alpha\in[0,1)$ is the largest real root of the degree 4 polynomial,
\begin{multline}
     f(\lambda) = \lambda^4  + (4- \alpha^2)\lambda^3  + \left(\frac{11}{4} \alpha^4 -12 \alpha^2 - 4\right) \lambda^2   \\
  +(2 \alpha^6 - \alpha^4 - 20 \alpha^2 -32 )\lambda  + 5 \alpha^6 - 21 \alpha^4 + 16 \alpha^2 - 32.
  \label{analySol}
 \end{multline}
 Moreover, $C_{\alpha,\alpha}(\bm{p})=c_{\mathcal{Q}}(\alpha,\alpha)$ \emph{self-tests} a two-qubit quantum strategy with optimal $(*)$ local observables of the form \eqref{qubitParam}, such that the optimal cosines are equal, i.e., ${c}^*(\alpha)=c^*_{A}(\alpha)=c^*_{B}(\alpha)$,
and satisfy,
 \begin{equation}
    {c}^*(\alpha) = \frac{1}{8}\left[ 3\alpha^2 - 4 + \sqrt{16+9\alpha^4+8\alpha^2(2c_{\QC}(\alpha,\alpha)-1)} \right]. \label{beta_c_relation}
\end{equation}
\end{theorem}
\begin{proof}
The proof relies on Lagrange multipliers to recast the problem as a system of nonlinear polynomial equations, which is then solved via elimination using Gröbner basis (see Appendix A of \cite{Lee2017}). The elimination step was carried out in Mathematica; for the convenience of the reader we provide a Mathematica notebook containing the proof \cite{scalaSelfTesting}.

Consider the parametrization presented in \eqref{qubitParam} and assume $\alpha=\beta$. Then the Bell operator in \eqref{sym_tilted_functional} has the following matrix representation
\begin{equation}\label{qubit_bell_matrix}
       \hat{C}_{\alpha,\alpha} = \left(
\begin{array}{cccc}
 \omega+2\alpha & s_B-c_A s_B & s_A-c_B s_A & -s_A s_B \\
 s_B-c_A s_B & -\omega & -s_A s_B & (c_B-1) s_A \\
 s_A-c_B s_A & -s_A s_B & -\omega & (c_A-1) s_B \\
 -s_A s_B & (c_B-1) s_A & (c_A-1) s_B & \omega-2\alpha \\
\end{array}
\right)
\end{equation}
where $\omega= 1+c_A+c_B-c_Bc_A$. The characteristic polynomial of \eqref{qubit_bell_matrix} has the expression,
\begin{multline}\label{characteristic_poly}
  q(\lambda, c_A, c_B, \alpha) = \lambda^4 - (4\alpha^2 + 8)\lambda^2 \\
    + 8\alpha^2(c_Ac_B - c_A - c_B -1)\lambda \\
    + 8[2(c_A^2+c_B^2 - c_A^2c_B^2) \\
    + \alpha^2(c_A^2c_B + c_Ac_B^2 - c_A^2 - c_B^2 - c_A - c_B)].
\end{multline} 

Note that parameters $s_A,s_B$ have been rewritten in terms of $c_A,c_B$, so that \eqref{characteristic_poly} is a polynomial function of the cosines. For the optimal quantum strategy, $c_A=c^*_A(\alpha)$, $c_B=c^*_B(\alpha)$, the maximal quantum value $c_{\QC}(\alpha,\alpha)$ corresponds to the maximum eigenvalue $\lambda^*(\alpha)$ of \eqref{qubit_bell_matrix}. Hence, to find the maximal quantum value $c_\QC(\alpha, \alpha)$ we need to maximize the largest root of \eqref{characteristic_poly} over the parameters $c_A,c_B\in[0,1)$. 

The Lagrangian for this optimization problem is
\begin{equation}
    L = \lambda + s\cdot q(\lambda, c_A, c_B, \alpha),
\end{equation}
where $s$ is a Lagrange multiplier. A stationary point of $L$ satisfies the conditions $\partial_{\lambda}L = \partial_{c_A}L = \partial_{c_B} L = \partial_s L = 0$, and is therefore a solution of the following polynomial equations,
\begin{subequations}\label{poly_conditions0}
\begin{gather} \label{poly_conditions1}
        \partial_{\lambda}L = 0, \\ \label{poly_conditions2}
        \partial_{c_A}q(\lambda, c_A, c_B, \alpha) = 0, \\ \label{poly_conditions3}
        \partial_{c_B}q(\lambda, c_A, c_B, \alpha) = 0, \\ \label{poly_conditions4}
        q(\lambda, c_A, c_B, \alpha) = 0.
\end{gather}
\end{subequations}
The cosines ${c}_A^*(\alpha),{c}_B^*(\alpha)$ parametrizing the optimal measurements, and the maximal quantum value $c_{\QC}(\alpha,\alpha)={\lambda}^*(\alpha)$, are thus common roots of the polynomials $\partial_{\lambda}L,\, \partial_{c_A}q,\, \partial_{c_B}q$, and $q$ in \eqref{poly_conditions0}. These polynomials define a ring, for which a Gröbner basis \cite{Lee2017} is computed, which in turn is used to eliminate $s, c_A$ and $c_B$. 

The elimination procedure results in a degree $6$ polynomial over $\lambda$. Two out of the six roots, $2(1+\alpha)$ and $2(1-\alpha)$, coincide with and remain below the classical bound for all $\alpha\in[0,1]$, respectively. Taking the quotient of the degree $6$ polynomial with respect to the product of these two roots results in $f(\lambda)$ in \eqref{analySol}.

To derive the optimal cosines $c_A^*(\alpha)$, $c_B^*(\alpha)$, we consider the equations \eqref{poly_conditions2}, \eqref{poly_conditions3}. Since, $c_B=1$, $c_A=1$ correspond to compatible measurements, we divide \eqref{poly_conditions2}, \eqref{poly_conditions3} by their respective factors $8(1-c_A)$, $8(1-c_B)$, to retrieve the following equations,
\begin{subequations}\label{poly_conditions}
\begin{gather} \label{poly_conditions23transformed2}
        (\alpha^2-4c_A)(1+c_B)+2\alpha^2 c_A + \alpha^2\lambda=0, \\ \label{poly_conditions23transformed3}
        (\alpha^2-4c_B)(1+c_A)+2\alpha^2 c_B + \alpha^2\lambda=0.  
\end{gather}
\end{subequations}
Taking the difference of \eqref{poly_conditions23transformed2} and \eqref{poly_conditions23transformed3} to eliminate $\lambda$, we retrieve the condition, $\left(\alpha ^2-4\right) (c_A-c_B)=0$. Since $\alpha\in[0,1)$, we conclude that the optimal cosines must be equal, i.e., $c_A^*(\alpha)=c_B^*(\alpha)=c^*(\alpha)$. Plugging $c_A=c_B=c$ back into \eqref{poly_conditions23transformed2}, we find that the optimal cosine $c^*(\alpha)$ must be a root of the following degree $2$ polynomial,  
\begin{equation} \label{finalPoly}
    h(c)=4c^2 + (4-3\alpha^2)c - \alpha^2(1+\lambda).
\end{equation}
We find that, for all $\alpha\in[0,1)$ and $\lambda=\lambda^*(\alpha)=c_{\QC}(\alpha,\alpha)$, the optimal cosine $c^*(\alpha)$ corresponds to the largest root of \eqref{finalPoly}, given by \eqref{beta_c_relation}, since the other root $c'=c_A=c_B\notin [0,1)$. Since the optimal cosine $c^*(\alpha
)$ \eqref{beta_c_relation} is uniquely determined by $c_{\QC}(\alpha,\alpha)$, we conclude that the optimal measurements are self-tested by the optimal quantum value $C_{\alpha,\alpha}(\bm{p})=c_{\QC}(\alpha,\alpha)$. Moreover, we find that the maximum eigenvalue of \eqref{qubit_bell_matrix} with optimal settings $c_A=c_B=c^*(\alpha)$ is nondegenerate (since no other eigenvalue violates the local bound $2(1+\alpha)$) for all $\alpha\in[0,1)$, which implies that $C_{\alpha,\alpha}(\bm{p})=c_{\QC}(\alpha,\alpha)$ also self-tests the optimal two-qubit non-maximally entangled state specified by the eigenvector associated with the maximum eigenvalue.  

\end{proof}
In FIG. \ref{fig:SelfTestingStatesAndMeas} we the plot the optimal cosine $c_A=c_B=c^*(\alpha)$ parametrizing the optimal measurements \eqref{qubitParam} and the Schimdt coefficient $\xi^*$ of the optimal partially entangled two-qubit state against $\alpha\in[0,1)$ self-tested by $C_{\alpha,\alpha}(\bm{p})c_{\QC}(\alpha,\alpha)$. The analogous self-testing statements for the general case of $\alpha\neq\beta$ are contained in Theorem 2, which, along with the analogous proof, has been deferred to the Supplementary Information for brevity. In particular, in contrast to the symmetric ($\alpha=\beta$) case for which a closed-form solution is presented in the Mathematica notebook \cite{scalaSelfTesting}, the maximal quantum value $c_{\QC}(\alpha,\beta)$ in the general case ($\alpha\neq \beta$) corresponds to the largest real root of a degree 6 polynomial \eqref{analySolgen}, which can be obtained numerically to any desired precision.

\begin{figure}[h] 
    \centering
     \includegraphics[width=\linewidth]{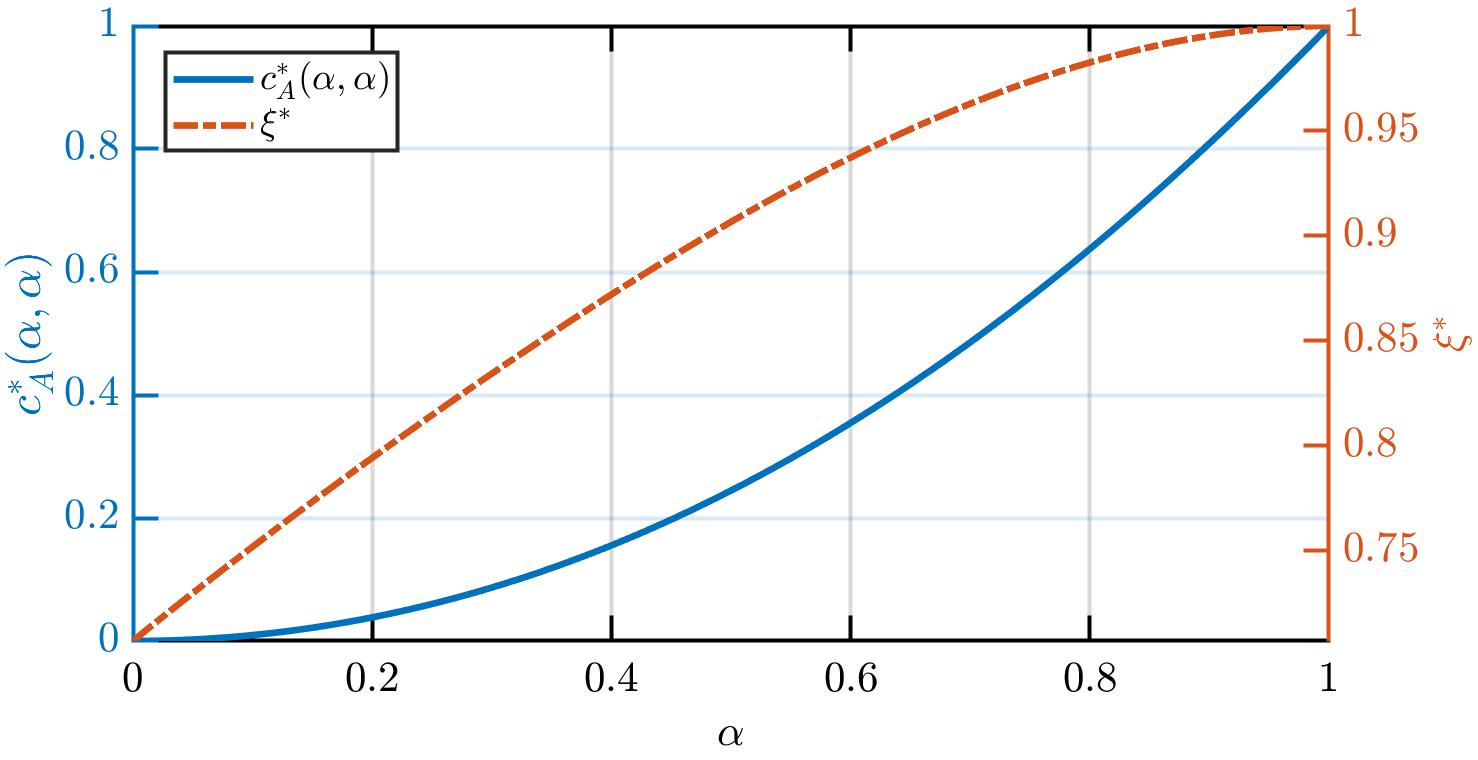}
    \caption{\emph{Self-testing of non-maximally incompatible measurements and non-maximally entangled state:---} Plots of $(i.)$ (solid blue line) the optimal cosine $c_A=c_B=c^*(\alpha)$ \eqref{qubitParam}, and $(ii.)$ (dashed orange line) the Schmidt coefficient $\xi^*$ of the optimal non-maximally entangled quantum state $\ket{\psi}=\xi^*\ket{00}+\sqrt{1-{\xi^*}^2}\ket{11}$ (represented in the Schmidt basis), against the tilting parameter $\alpha\in[0,1)$, self-tested by the maximal quantum value $C_{\alpha,\alpha}(\bm{p})=c_{\mathcal{Q}}(\alpha,\alpha)$ of the symmetrically ($\alpha=\beta$) tilted CHSH inequality \eqref{general_tilted}. As $\alpha\to 1$ ($\eta \to \frac{2}{3}$), the optimal qubit measurements as well as the optimal two-qubit entangled state become almost compatible and product, respectively.}
    \label{fig:SelfTestingStatesAndMeas}
\end{figure}

\subsection{Robust self-testing}

To demonstrate the robustness of the self-testing statements in Theorem \ref{selfTestSymmetric}, we use the numerical SWAP method introduced in \cite{bancal2015nswap, yang2014nswap}. The numerical SWAP technique utilizes the NPA hierarchy to obtain lower bounds $\FC^*_{L}\leq \FC^*\leq \FC$ on the fidelity $\FC$ between the state $\ket{\psi'}$ in a given quantum strategy and the reference state $\ket{\psi}$ which is self-tested, say by the maximal quantum violation $\beta(\bm{p})=\beta_{\QC}$ of a Bell inequality \eqref{general_bell_inequality}, where $L$ is level of NPA hierarchy. It builds upon the fact that, in case of self-testing, the local isometries of the form \eqref{isometries} mapping an optimal state to the target state $\ket{\psi}$ can be implemented via a partial SWAP unitary $U_{SWAP}$ which depends on the optimal measurements. It follows then that every state in an optimal quantum strategy, and only those states, saturate the inequality $\FC = \bra{\psi}U_{SWAP}\ket{\psi'} \leq 1$. Hence, the problem of certifying that the maximal quantum value of a Bell functional self-tests the target state can be rephrased as that of checking whether the minimal fidelity $\FC^*$, over optimal realizations, is $1$. The advantage in this recasting of the problem is that the finding the minimal fidelity $\FC^*$ can be relaxed via the NPA hierarchy.

Another advantage of this method is that it can be used to numerically study the robustness of the self-testing results. If minimization of $\FC$ is carried over quantum strategies attaining a suboptimal functional value $c_\QC(\alpha, \alpha)-\epsilon$, then the minimum fidelity $\mathcal{F}^*$ quantifies how close the state in any such strategy must be to the target $\ket{\psi}$. In figure FIG.~\ref{fig:numerical_swap} we plot the numerically computed lower bounds on minimum fidelity $\mathcal{F}^*_{L}\leq \FC^*$ from level $L$ of the NPA hierarchy against the value of the symmetrically tilted functional $C_{\alpha,\alpha}(\bm{p})$, for $\alpha\in\{0.2,0.3,0.4,0.6,0.8\}$. Since the optimal observables $\hat{A}_1$ and $\hat{B}_1$ are given by some linear combination of gates $\hat{Z}$ and $\hat{X}$ in the SWAP circuit \eqref{qubitParam}, to find $\FC^*_L$, we introduce extra dichotomic operators $\hat{A}_2$ and $\hat{B}_2$ to account for the $X$ gate on Alice's and Bob's side and impose the relevant extra constraints by means of localising matrices \cite{bancal2015nswap}.
\begin{figure}[h]
    \centering
     \includegraphics[width=\linewidth]{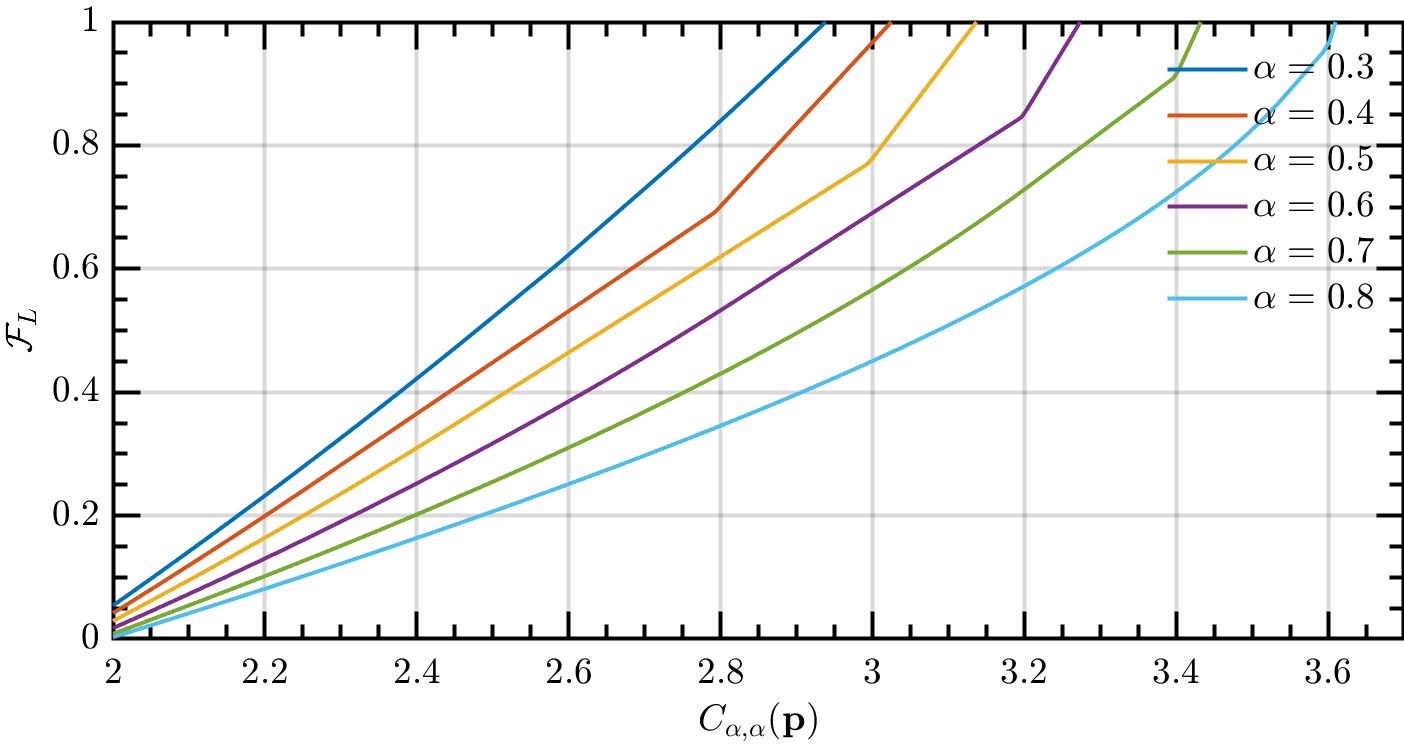}
    \caption{\emph{Robustness of the self-testing statements:---} A plot of lower bounds $\mathcal{F}^*_L$ on the minimum quantum fidelity $\mathcal{F}_L^*\leq \mathcal{F}^*$ from the level $L=3$ of NPA hierarchy between the actual state and the optimal self-testing state against the observed value $C_{\alpha,\alpha}(\bm{p})$ of the symmetrically ($\alpha=\beta$) tilted CHSH functional \eqref{general_tilted} for tilting parameters $\alpha\in \{0.3, 0.4, 0.5,0.6,0.7,0.8\}$.} 
    \label{fig:numerical_swap}
\end{figure}
We find that level $3$ of the NPA hierarchy is enough for producing the fidelity curves for $\alpha < 0.9$, but higher values of $\alpha$ require increasing levels of the hierarchy. This rise in the minimum hierarchy level required for certification is also observed in the calculation of NPA upper bounds for the maximal quantum value $ c_{\QC}(\alpha,\alpha)\leq c_{\QC_L}(\alpha,\alpha)$, which we address in the following subsection. As can be seen in the figure, the minimal fidelity becomes $\mathcal{F}^*_L = \mathcal{F}^* = 1$ when the functional value is maximal, as expected from the discussion above. In particular, for $\alpha>0.3$ a clear change in behavior in the fidelity curve is observed when the functional value reaches the local bound $C_{\alpha,\alpha}(\bm{p}) = c_{\LC}(\alpha,\alpha)=2(1+\alpha)$, after which the dependence is almost linear.  

\subsection{Analytical intractability of SOS decompositions}
\label{SOSintract}
As already mentioned above, the SOS decompositions found in the case of one ideal detector do not easily generalize to the general case of efficiencies $\eta_A, \eta_B < 1$. On one hand, failed numerical attempts to find tight SOS decompositions with degree $2$ polynomials suggested they were not available for some values of the tilting parameters $\alpha, \beta$. Since a polynomial-degree SDP relaxation of the SOS problem is dual to the moment SDP relaxation introduced by NPA, we started by charting, in the tilting parameter space $\alpha,\beta\in[0,2]$, the minimum level $L$ required for the NPA upper bound $c_{\QC_L}(\alpha,\beta)$ to be tight such that $c_{\QC_L}(\alpha,\beta)=c_{\QC}(\alpha,\beta)$, where the maximal quantum values $c_{\QC}(\alpha,\beta)$ of the doubly-tilted CHSH inequalities are obtained from Theorem \ref{selfTestSymmetric} and 2 up to any desired precision. We show the results of this numerical exploration in In FIG. \ref{NPAlevels}. 
\begin{figure}[H] 
    \centering
     \includegraphics[width=\linewidth]{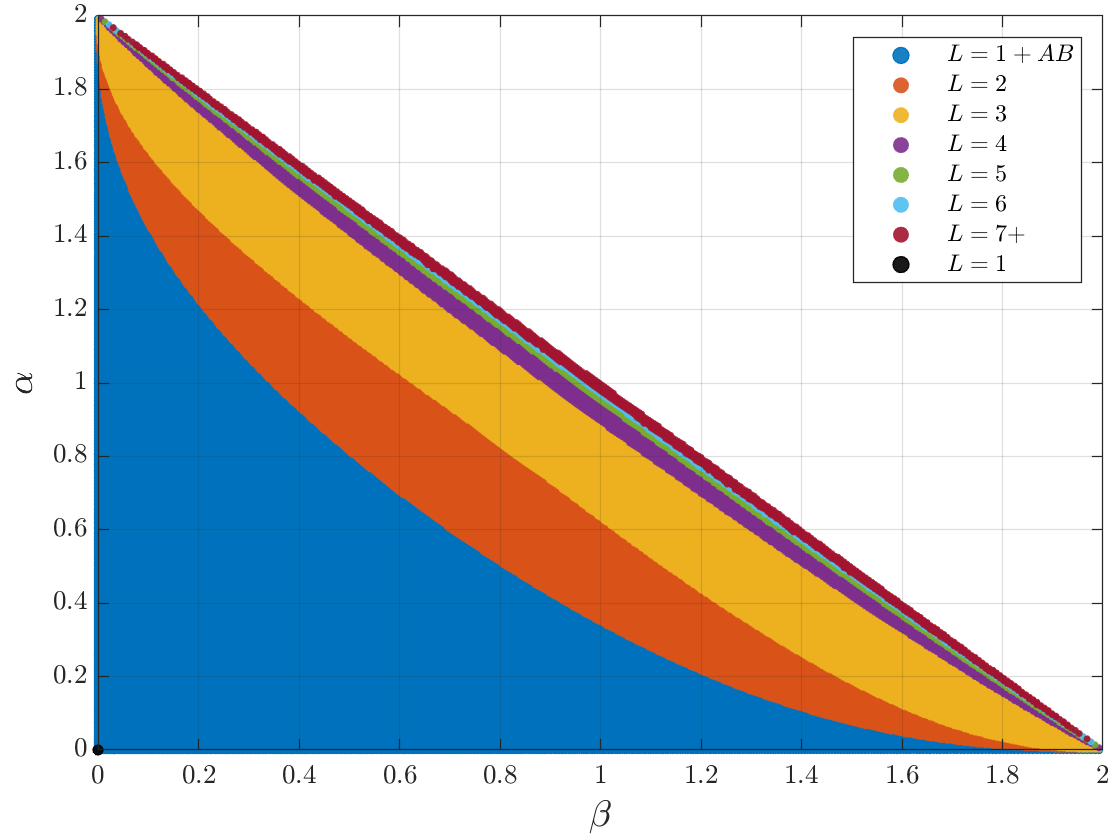}
    \caption{\label{NPAlevels} \emph{ Increasing levels for tight NPA upper bounds in the CHSH scenario:---}  
    A plot of the minimum level $L\in\{1,1+AB,2,3,4,5,6,7+\}$ of the NPA hierarchy required to saturate the maximum quantum value $c_{\QC}(\alpha,\beta)$ of the doubly-tilted CHSH inequality \eqref{general_tilted} such that $c_{\QC_L}(\alpha,\beta)=c_{\QC}(\alpha,\beta)$ for $\alpha,\beta\in[0,2]$. Notably, while level $1+AB$ suffices when either $\alpha=0$ (x axes) or $\beta=0$ (y axes), the required minimum level of the NPA hierarchy rapidly increases as the tilting parameters tend toward the critical boundary $\alpha + \beta = 2$.}
\end{figure}
As the figure shows, there is a rapid increase of $L$ as the tilting parameters approach the critical boundary $\alpha+\beta=2$. In particular, for $\alpha=\beta=0.999$ we find that not even $L=10$ is enough for the NPA upper bound to match the maximum quantum value, as illustrated by the upper bound sequence in FIG.~\ref{NPAlevels1}. In particular, to a $12$-digits approximation the analytical answer is $c_{\QC}(\alpha=0.999,\beta=0.999)=3.9980\ 0000\ 1333$, whereas the NPA hierarchy gives the upper bound of $c_{\QC_{L=10}}(\alpha=0.999,\beta=0.999)=3.9980\ 0000\ 2190$. These high precision calculations were performed using the toolkit for non-commutative polynomial optimization Moment \cite{garner24}, the modeller YALMIP \cite{yalmip}, and the arbitrary-precision solver SDPA-GMP \cite{Nakata2010}. 

These observations raise the question of whether there is a finite level at which the NPA upper bound coincides with the maximum quantum value for all tilting parameters $\alpha, \beta$. Because of the SDP duality mentioned earlier, it is easy to see that NPA level $L$ corresponds to polynomial degree on the SDP relaxation of the SOS decompositions problem. Thus the problem of finding tight SOS decompositions (of the form \eqref{SOS}) for the doubly-tilted CHSH inequalities \eqref{general_tilted} is intractable via analytical methods described in \cite{tiltedasym} for the tilted CHSH inequalities \eqref{tiltedasym}.

\begin{figure}
    \centering
    
     \includegraphics[width=1\linewidth]{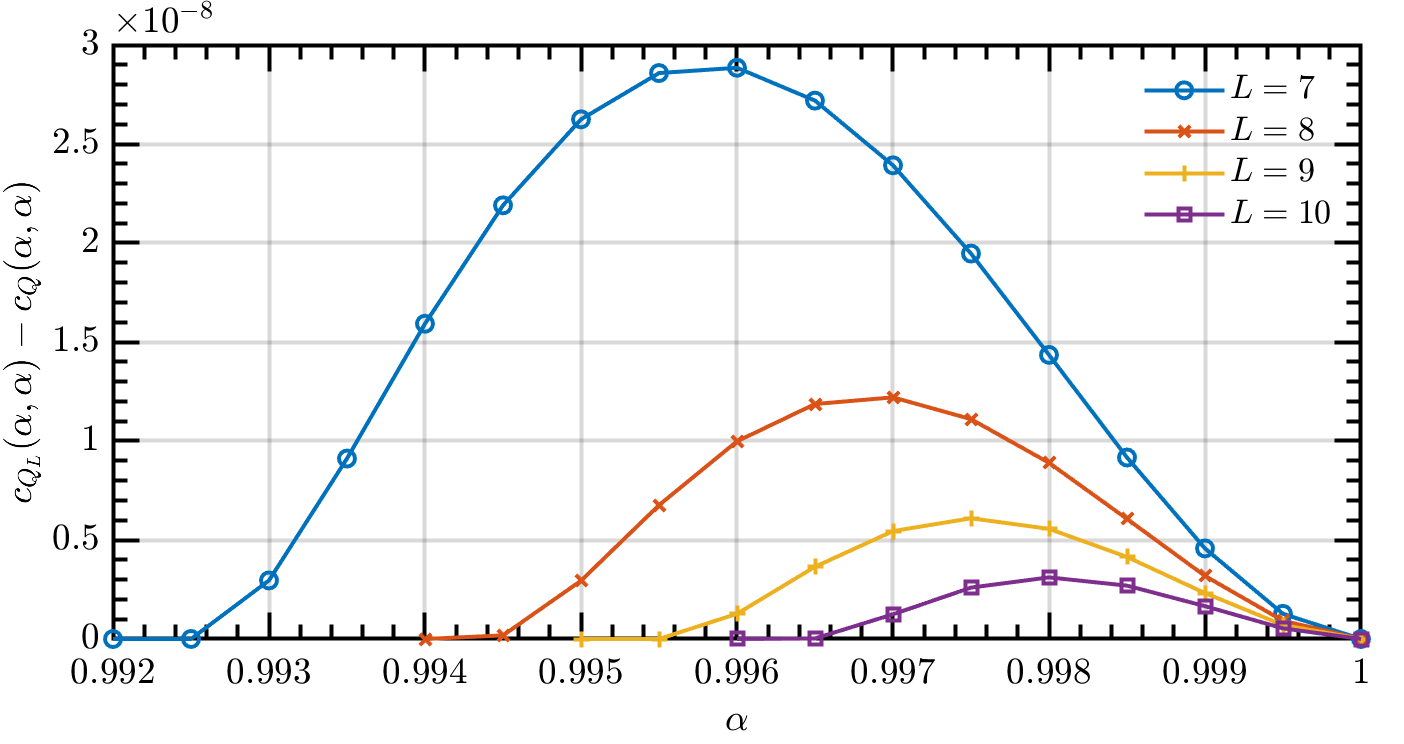}  \caption{\label{NPAlevels1} \emph{ Error in estimation of maximum quantum violation with high levels of NPA hierarchy:---} The curves in plot (a) correspond to the difference  $c_{\QC_L}(\alpha,\alpha)-c_{\QC}(\alpha,\alpha)$ ($\times10^{-8}$) between the upper-bounds from the levels $L\in\{7,8,9,10\}$ of the NPA hierarchy and the maximal quantum violation $c_{\QC}(\alpha,\alpha)$ of symmetrically ($\alpha=\beta$) tilted Bell inequalities \eqref{general_tilted}, against $\alpha\in[0.99,1]$, obtained with an arbitrary precision solver SDPA-GMP \cite{Nakata2010} and our analytical solution (Theorem \ref{selfTestSymmetric}), respectively. }
\end{figure}
\section*{Discussions}

In this work, we addressed the problem of finding the quantum strategies that maximize the loophole-free violation of a given Bell inequality in the presence of inefficient detectors. In Lemma \ref{tiltedTheorem}, we demonstrate that for \emph{any} Bell inequality and \emph{any} specification of detection efficiencies, the quantum strategies that yield the maximal loophole-free violation are the ones that maximally violate a tilted version of the Bell inequality. We then consider the CHSH inequality \eqref{CHSH_ineq} to retrieve a family of doubly-tilted CHSH inequalities \eqref{tilted_CHSH} from Lemma \ref{tiltedTheorem}. As our main result, we analytically derive self-testing statements (Theorem \ref{selfTestSymmetric} and Theorem 2) for the doubly-tilted CHSH inequalities, entailing the maximal quantum violation and the ensuing optimal quantum strategy. We note that the maximal quantum violation of the symmetrically ($\alpha,\beta$) tilted CHSH inequality \eqref{general_tilted} from  Theorem \ref{selfTestSymmetric} differs from that reported in \cite{miklin2022exponentially}. Additionally, it is worth noting that the nonlocal correlations maximally violating the symmetrically $(\alpha=\beta)$ tilted inequalities in \eqref{general_tilted} also maximize the sum $\expval{A_0} + \expval{B_0}$ for a given value $C(\bm{p})$ of the CHSH functional. Consequently, the quantum strategies obtained via Theorem \ref{selfTestSymmetric} allow us to recover the boundary of the set of quantum correlations on the slice $C(\bm{p})$ vs $\expval{A_0}+\expval{B_0}$ of the no-signaling polytope, plotted in FIG. \ref{fig:anufig}. Furthermore, these quantum strategies have found application in the recently introduced routed Bell experiments \cite{miklin2022exponentially,chaturvedi2024extending,Lobo2024certifyinglongrange}.

Besides providing a convenient way for finding quantum strategies that generate the maximum loophole-free nonlocality, Lemma \ref{tiltedTheorem}, and in particular the expression of the tilted Bell inequalities \eqref{tilted_B}, offers crucial insights into how the optimal quantum behaviors move with the efficiencies of the detectors, in the no-signaling polytope. 
Essentially, with decreasing efficiencies, the hyperplane corresponding to the Bell inequality in \eqref{tilted_B} tilts about a local deterministic point specified by the assignment strategy, in turn moving the optimal strategy along towards the local polytope $\LC$, and specifically towards the local deterministic point. We exemplify this observation in FIG. \ref{fig:anufig}.

The family of self-testing statements in Theorem \ref{selfTestSymmetric} and 2 provide crucial insights into how decreasing detector efficiency affects optimal quantum strategies. For inefficient detectors $\eta_A,\eta_B<1$, the optimal strategy requires partially entangled states for maximal loophole-free nonlocality and the degree of entanglement decreases as the efficiencies approach the critical values \eqref{critical_efficiency}, as the state becomes almost product (see FIG. \ref{fig:SelfTestingStatesAndMeas} for the symmetric ($\alpha=\beta$) case). While this finding is in line with observations from the known asymmetric case $\eta_A=1$, the optimal measurements for the general case $\eta_A\neq \eta_B$ present an intriguing deviation. Specifically, in the asymmetric case $\eta_A=1$ ($\beta=0$), Alice's optimal measurements \eqref{measurements_asymmetric} remain maximally incompatible irrespective of Bob's detection efficiency $\eta_B \in (1/2,1]$ ($\alpha\in[0,2)$), while Bob requires partially incompatible measurements whose incompatibility decreases with his decreasing efficiency $\eta_B$ ($\alpha$), approaching almost compatible measurements as $\eta_B\to1/2$ ($\alpha\to 1$). However, the situation changes significantly when both detectors are inefficient. As we illustrate in FIG. \ref{fig:SelfTestTheM}, in contrast to the asymmetric case, Alice's measurements depend non-trivially on Bob's detection efficiency $\eta_B$. Interestingly, whenever $\eta_A<1$ ($\beta>0$), Alice's optimal measurements sharply tend towards  compatible as $\eta_B$ ($\alpha$) approaches the critical boundary ($\alpha\to 2-\beta$) \eqref{critical_efficiency}. This observation highlights the natural importance of partially incompatible measurements in device independent cryptography with imperfect detectors \cite{masini2024jointmeasurability}. 

\begin{figure} 
    \centering
   \includegraphics[width=\linewidth]{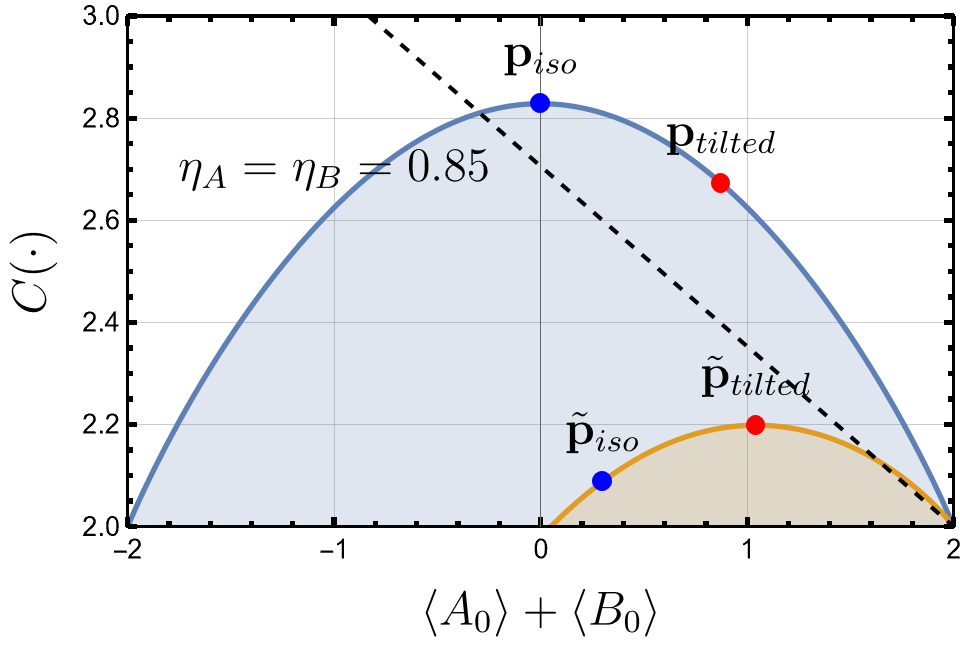}
    \caption{\emph{Effect of inefficient detectors on nonlocal correlations:---} This graphic illustrates the impact of detector inefficiencies on nonlocal quantum correlations within the simplest Bell scenario. 
    The blue region represents the set of quantum correlations $\bm{p}\in\QC$. The solid blue curve represents the boundary of $\QC$ on the slice $C(\bm{p})$ vs $\mbraket{A_0}+\mbraket{B_0}$ of the no-signaling polytope.
    With the detection efficiencies $\eta_A=\eta_B=0.85$, and the local assignement strategy $q(ab|xy)=\delta_{a,+1}\delta_{b,+1}$, the effective quantum correlations $\tilde{\bm{p}}$ \eqref{effectiveP} are constrained to the smaller orange subset. The blue dot on the solid blue curve corresponds to the isotropic behavior $\bm{p}_{iso}$ that maximally violates the CHSH inequality \eqref{CHSH_ineq}, $C(\bm{p}_{iso})=2\sqrt{2}$, in ideal conditions, while the corresponding effective behavior (blue dot on the solid orange curve) $\tilde{\bm{p}}_{iso}$ no longer attains the maximum loophole-free violation of the CHSH inequality,  $C(\tilde{\bm{p}}_{iso})\approx 2.08854$. The red dot on the solid blue curve corresponds to the quantum behavior $\bm{p}_{tilted}$ which maximally violates the doubly-tilted CHSH inequality \eqref{tilted_CHSH} (dashed black line), $C_{\eta_A,\eta_B}(\bm{p}_{tilted})=2.98098$. The corresponding effective behavior (red dot on the solid orange curve) $\bm{p}_{tilted}$
    attains the maximum loophole-free violation $C(\tilde{\bm{p}}_{tilted})\approx2.19876 $ of the CHSH inequality, thereby, exemplifying Lemma \ref{tiltedTheorem}.
    }
    \label{fig:anufig}
\end{figure}

Besides the inferences concerning the maximal effective nonlocality in the presence of inefficient detectors, Theorem \ref{selfTestSymmetric} and \ref{doublyTiltedTheorem} allow us to reveal a fascinating intricacy. While the NPA hierarchy is claimed to not converge to the set of quantum correlations in general \cite{ji2000}, lower levels (e.g. $L=1+AB,2,3$) were widely believed to be sufficient to characterize it in the CHSH scenario.
However, in striking contrast to the widely held belief, we demonstrate in FIG. \ref{NPAlevels} and \ref{NPAlevels1}, as the tilting parameters $\alpha,\beta>0$ approach critical limit $\alpha+\beta\to 2$, the level $L$ of NPA hierarchy required for a tight upper bound on the maximal quantum value of doubly-tilted CHSH functional \eqref{general_tilted} increases drastically (FIG. \ref{NPAlevels}). This effect is the most pronounced for the symmetric case $\alpha=\beta$ where we find that even level $10$ of the NPA hierarchy does not yield tight upper bounds on the maximal violation of the symmetrically tilted CHSH inequalities as $\alpha\to 1$ (FIG. \ref{NPAlevels1}). Crucially, it remains unclear if \emph{any} finite level of NPA will be enough to characterize all extremal quantum correlations even in the CHSH scenario. This effect also renders the traditional SOS decomposition method impractical for deriving analytical self-testing statements intractable. 

Even more strikingly, this unexpected feature is observed for quantum behaviors maximally violating the doubly-tilted CHSH inequalities \eqref{general_tilted} near the critical boundary $\alpha+\beta=2$, as depicted in FIG.~\ref{NPAlevels}, which are require with almost compatible measurements (FIG. \ref{fig:SelfTestTheM}) and almost product states. 
Since the level $1+AB$ is enough for the asymmetric case $\beta=0$ wherein Alice's optimal measurement remains maximally incompatible, irrespective of the value of $\alpha$, the effect of increasing NPA levels seems to be linked to the aforementioned intricate dependency of the optimal partially incompatible measurements on the tilting parameters $0<\alpha,\beta<2$, and which becomes sharper as the optimal measurements for both parties become almost compatible as $\alpha+\beta\to 2$. These results then raise the question of whether the complexity of the characterization of extremal nonlocal quantum correlation, as measured by the minimum level of the NPA hierarchy required to saturate the maximal quantum violation of the associated Bell inequality, is related to the partial incompatibility of the quantum measurements realizing them.

\begin{figure} 
    \centering
     \includegraphics[width=\linewidth]{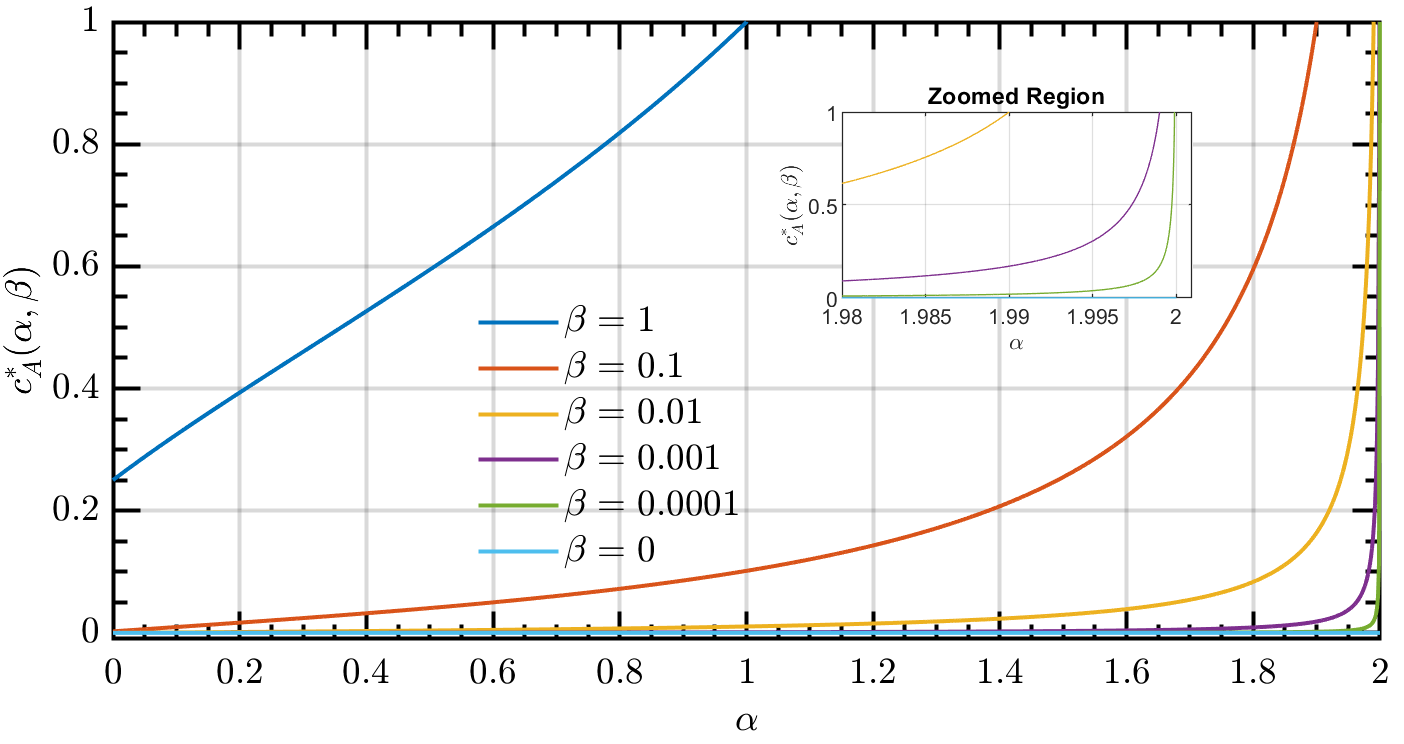}
    \caption{\emph{Self-testing of partially incompatible observables:---} A plot of the optimal cosines of Alice $c^*_A(\alpha,\beta)$ \eqref{qubitParam} with $\beta\in\{\alpha,0.1,0.0
    1,0.001,0\}$ self-tested by the maximum quantum value $C_{\alpha,\beta}(\bm{p})=c_{\QC}(\alpha,\beta)$ of the doubly-tilted CHSH inequalities \eqref{general_tilted} against the tilting parameter $\alpha\in[0,1]$. Here, $c^*_A(\alpha,\beta)=0$ implies maximally incompatible whereas $c^*_A(\alpha,\beta)=1$ reflects compatible observables. Notice, that in contrast to the asymmetrically tilted case $\beta=0$ wherein Alice's optimal cosine $c^*_A(\alpha,\beta=0)$ stays constant with respect to $\alpha$ \cite{tiltedasym}, for the general case, whenever $\beta>0$, Alice's optimal measurements change with $\alpha=\frac{2}{\eta_B}(1-\eta_B)$ and in-turn depend on Bob's detection efficiency $\eta_B$, and tend towards compatible measurements as $\alpha\to 2-\beta$.}
    \label{fig:SelfTestTheM}
\end{figure}

 \section*{Data Availability}
Data Availability does not apply to this manuscript.
 \section*{Code Availability}
The Mathematica notebooks illustrating the proofs of Theorem \ref{selfTestSymmetric}, Theorem \ref{doublyTiltedTheorem} and Lemma \ref{differentlyTiltedTheorem} are available online \cite{scalaSelfTesting}.
\section*{Acknowledgements} 
We would like to thank Tamás Vértesi, Marcin Pawłowski, Jedrzej Kaniewski,  Cosmo Lupo, Marcin Wieśniak and Marek Żukowski for insightful discussions. N.G. acknowledges support from CONICET of Argentina, CONICET PIP Grant No. 11220200101877CO and National Science Centre, Poland under the SONATA project ``Fundamental aspects of the quantum set of correlations'' (grant no.~2019/35/D/ST2/02014). G.S. received funding from the European Union’s Horizon Europe research and innovation program under the project "Quantum Secure Networks Partnership" (QSNP, grant agreement No 101114043), QuantERA/2/2020, an ERA-Net co-fund in Quantum Technologies, under the eDICT project, and INFN through the project "QUANTUM". The research of M.A. was supported by the European Union--Next Generation UE/MICIU/Plan de Recuperación, Transformación y Resiliencia/Junta de Castilla y León.  E.P. acknowledges support by NCN SONATA-BIS grant No. 2017/26/E/ST2/01008.
A.C. acknowledges financial support by NCN grant SONATINA 6 (contract No. UMO-2022/44/C/ST2/00081).

\section*{Competing Interests}
The authors declare that there are no competing interests.
\section*{Author Contributions}
All authors contributed equally to this manuscript.

\bibliography{cite,bib_loopholes}
\appendix
\clearpage

\section{Proof of Lemma 1} \label{proofOfLemma1}
In this section we present the detailed proof for Lemma \ref{tiltedTheorem}. In a $(m_A,m_B,d_A,d_B)$-Bell scenario, consider an ideal behavior $\bm{p} \in \mathbb{R}_+^{m_A d_A m_B d_B}$ which violates a given Bell inequality $\beta(\bm{p})\leq \beta_\LC$ \eqref{general_bell_inequality}. Let us now consider imperfect detectors. In particular, the detectors sometimes fail to click resulting in a ``no-click" event. Without loss of generality, we may treat the ``no-click" as an additional local outcome, denoted by $\varnothing$. To model the effect of inefficient detectors on the observed behavior, we assume that Alice's and Bob's detectors click independently of each other and independently of the ideal behavior such that $p_A(\varnothing|x)=\eta_A,p_B(\varnothing|y)=\eta_B$ for all $x,y$. We note that this assumption is purely operational and concerns the observed behavior and does not restrict in any way the internal workings of the devices, which may be in control of a eavesdropper. Then effect of imperfect detectors is described by an affine map $\Omega^{(\varnothing)}_{\eta_A,\eta_B}:\mathbb{R}^{m_A d_A m_B d_B}_+\to\mathbb{R}^{m_A (d_A+1) m_B (d_B+1)}_+$. For any $\eta_A,\eta_B$: $\Omega^{(\varnothing)}_{\eta_A,\eta_B}$ transforms an ideal behavior $\bm{p}$ from a $(m_A,m_B,d_A,d_B)$-Bell scenario into an effective behavior $\bar{\bm{p}}=\Omega^{(\varnothing)}_{\eta_A,\eta_B}(\bm{p})$ in a $(m_A,m_B,d_A+1,d_B+1)$-Bell scenario such that the entries of $\bar{\bm p}$ are given by,
\begin{align}\label{detE}
    \bar{p}(ab|xy)  =  \begin{cases} \eta_A\eta_Bp(ab|xy), \ \ \text{if} \ a\neq \varnothing \ \& \ b\neq \varnothing, \\
    (1-\eta_A)\eta_B p_B(b|y), \ \ \text{if} \ a= \varnothing \ \& \ b\neq \varnothing, \\
    \eta_A(1-\eta_B) p_A(a|x), \ \ \text{if} \ a\neq \varnothing \ \& \ b= \varnothing, \\
    (1-\eta_A)(1-\eta_B), \ \ \text{if} \ a= \varnothing \ \& \ b= \varnothing.
    \end{cases}
\end{align}
An effective behavior $\bar{\bm{p}}$ is nonlocal in a loophole-free way if and only if it lies outside the local polytope $\LC$ of an enlarged $(m_A,m_B,d_A+1,d_B+1)$-Bell scenario. 

Keeping $\varnothing$ as an additional outcome has two drawbacks. Besides increasing the complexity of the characterization of the correlations sets $\LC,\QC$, this strategy is experimentally cumbersome as it requires $d_A+d_B$ detectors. A more connivent approach which preserves the Bell scenario and requires $d_A+d_B-2$ detectors comprises of locally assigning a pre-existing outcome to the ``no-click" event. We now describe assignment strategies and their effect. 

Let the parties employ a local assignment strategy to mitigate the ``no-click" events. In case of a ``no-click" event Alice and Bob query an assignment strategy specified a vector $\bm{q}\in \mathbb{R}^{m_A d_A m_B d_B}_+$ whose entries $q(ab|xy)$ specify the probabilities of producing the outcomes $a,b$ conditioned on the inputs $x,y$, respectively. 
The effect of such an assignment strategy then can be summarized as an affine map $\zeta^{\varnothing\mapsto \bm{q}}:\mathbb{R}^{m_A (d_A+1) m_B (d_B+1)}_+\mapsto \mathbb{R}^{m_A d_A m_B d_B}_+$. $\zeta^{\varnothing\mapsto \bm{q}}$ transforms a behavior $\bar{\bm{p}}$ from the larger $(m_A,m_B,d_A+1,d_B+1)$-Bell scenario to a behavior $\tilde{\bm{p}}=\zeta^{\varnothing\mapsto \bm{q}}(\bar{\bm{p}})$ in the original $(m_A,m_B,d_A,d_B)$-Bell scenario, such that the entries of $\bm{\tilde{p}}$ are given by,
\begin{align} \label{assN}\nonumber
    \tilde{p}(ab|xy)&=\bar{p}(ab|xy) + q_A(a|x)\bar{p}(a=\varnothing b|xy)\\  \nonumber
    & \ \ \ \ + q_B(b|y)\bar{p}(ab=\varnothing|xy) \\ 
     & \ \ \ \ +q(ab|xy)p(a=\varnothing b=\varnothing|xy), 
\end{align}
where we have implicitly assumed that $\bm{q}$ is no-signaling such that the marginal distributions $q_A(a|x)=\sum_a q(ab|xy),q_B(b|y)=\sum_b q(ab|xy)$, are well defined. 
 The combined effect of $\Omega^{(\varnothing)}_{\eta_A,\eta_B}$ followed by $\zeta^{\varnothing\mapsto \bm{q}}$, on the ideal behavior $\bm{p}$ can be summarized as an affine map 
$\Omega^{(\varnothing \mapsto \bm{q})}_{\eta_A,\eta_B}:\mathbb{R}^{m_A d_A m_B d_B}_+\mapsto \mathbb{R}^{m_A d_A m_B d_B}_+$. For any $\eta_A,\eta_B$: $\Omega^{(\varnothing)\mapsto \bm{q}}_{\eta_A,\eta_B}$ transforms an ideal behavior $\bm{p}$ in a $(m_A,m_B,d_A,d_B)$-Bell scenario into an effective behavior $\tilde{p}=\Omega^{(\varnothing)\mapsto \bm{q}}_{\eta_A,\eta_B}(\bm{p})=\zeta^{\varnothing\mapsto \bm{q}}(\Omega^{(\varnothing)}_{\eta_A,\eta_B}(\bm{p}))$ in the same $(m_A,m_B,d_A,d_B)$-Bell scenario. Plugging in \eqref{detE} into \eqref{assN} we retrieve the entries of effective behavior $\tilde{\bm{p}}$,
\begin{align} \label{assN2}\nonumber
    \tilde{p}(ab|xy)&= \eta_A\eta_Bp(ab|xy)+(1-\eta_A)\eta_B q_A(a|x)p_B(b|y)\\ \nonumber
    & \ \ \ \ + (1-\eta_B)\eta_A p_A(a|x)q_B(b|y) \\
    & \ \ \ \ +(1-\eta_A)(1-\eta_B)q(ab|xy). 
\end{align}
Since \eqref{assN2} holds for all $x,y,a,b$, it can be succinctly summarized in terms of ideal behavior vectors, $\bm{p}\in\mathbb{R}^{m_A d_A m_B d_B}_+ ,\bm{p}_A\in \mathbb{R}^{m_A d_A}_+,\bm{p}_B\in\mathbb{R}^{m_B d_B}_+$, and the behavior vectors associated with the assignment strategy, $\bm{q}\in\mathbb{R}^{m_A d_A m_B d_B}_+,\bm{q}_A\in \mathbb{R}^{m_A d_A}_+,\bm{q}_B\in\mathbb{R}^{m_B d_B}_+$
\eqref{effectiveP},
\begin{align} \label{effectivePOmega} \nonumber
    \tilde{\bm p} & = \eta_A\eta_B{\bm p} + \eta_A(1-\eta_B){\bm p}^A\otimes {\bm q}^A \\  
    & \ \ \ \ + (1-\eta_A)\eta_B {\bm q}^A \otimes {\bm p}^B + (1-\eta_A)(1-\eta_B)\bm q.
\end{align} 
We are interested in the maximal loophole-free violation of the Bell inequality \eqref{lfviolation}, $\beta(\tilde{\bm{p}})-\beta_{\LC}$, where $\beta(\tilde{\bm{p}})$ has the following form,
\begin{align} \label{impEq} \nonumber
    \beta(\tilde{\bm{p}})=&\eta_A\eta_B \beta(\bm{p}) + (1-\eta_A)\eta_B\beta(\bm{q}_A\otimes\bm{p}_B)   \\ \nonumber
    &+ \eta_A(1-\eta_B)\beta(\bm{p}_A\otimes\bm{q}_B) \\
    &+(1-\eta_A)(1-\eta_B)\beta(\bm{q}),
    \end{align}
where we have used the convex-linearity of the Bell functional \eqref{general_bell_inequality}. 

Notice that for all $\eta_A,\eta_B\in[0,1]$, $\beta(\tilde{p})$ \eqref{impEq} is a convex combination four terms $\beta(\bm{p}),\beta(\bm{q}_A\otimes\bm{p}_B),\beta(\bm{p}_A\otimes\bm{q}_B),\beta(\bm{q})$. Since $\beta(\bm{p})\leq \beta_\LC$ for all $\bm{p}\in \LC$, and the assignment strategy is assumed to be local such that $\bm{q}\in \LC$, these terms are individually upper bounded by $\beta_\LC$. Hence, for any $\eta_A,\eta_B\in[0,1]$, and any local assignment strategy $\bm{q}\in \LC$, a violation of the given Bell inequality $\beta(\tilde{p})\leq \beta_\LC$ by the effective behavior $\tilde{\bm{p}}$ remains a sufficient condition for loophole-free nonlocality.

Since $\tilde{\bm{p}}$ \eqref{impEq} is a function of both the local the assignment strategy $\bm{q}\in\LC$ and the ideal quantum behavior $\bm{p}\in\QC$, finding its maximum value is a two-fold optimization problem. We first consider the optimization over local assignments $\bm{q}\in\LC$ strategies. Observe that for any fixed ideal behavior $\bm{p}\in\QC$ and any given $\eta_A,\eta_B$, $\tilde{\bm{p}}$ \eqref{impEq} forms a linear functional of the local assignment strategy $\bm{q}\in \LC$. Hence, to find its maximum value we need only consider the $(d_A)^{m_A}(d_B)^{m_B}$ deterministic vertices of the local polytope $\LC$ which correspond to deterministic assignment strategies. 

We denote a deterministic assignment by $\bar{\bm{q}}\in\mathbb{R}^{m_A d_A m_B d_B}_+$ with entries $\bar{q}(ab|xy)=\delta_{a,a_x}\delta_{b,b_y}$, where $a_x,b_y$ denote the local deterministic assignments. For each such deterministic assignment strategy, $\beta(\tilde{\bm{p}})$ \eqref{impEq} is a linear functional to be maximized over ideal quantum behaviors $\bm{p}\in\QC$. Then for any $\eta_A,\eta_B\in[0,1]$, the maximal loophole-free violation of Bell inequality is the maximum over the resultant set of $(d_A)^{m_A}(d_B)^{m_B}$ maximum quantum values $\beta(\tilde{\bm{p}})$ each associated with a deterministic assignment strategy. 

Maximizing the value of $\beta(\tilde{\bm{p}})$ \eqref{impEq} for any given deterministic strategy $\bar{\bm{q}}$ and $\eta_A,\eta_B$ amounts to finding the maximum quantum violation of a tilted version of a Bell inequality,
\begin{equation}
    \beta_{\eta_A,\eta_B}(\bm{p})  \leq  \beta_\LC(\eta_A,\eta_B) = \frac{\beta_\LC}{\eta_A\eta_B}-\frac{(1-\eta_A)(1-\eta_B)}{\eta_A\eta_B}\beta(\bar{\bm{q}}), 
\end{equation}
where the tilted Bell functional $\beta_{\eta_A,\eta_B}(\bm{p})$ is obtained by dividing \eqref{impEq} by $\eta_A\eta_B$ and rearranging,
\begin{align} \label{impEqF} \nonumber
      \beta_{\eta_A,\eta_B}(\bm{p})= & \beta(\bm{p}) + \frac{1-\eta_A}{\eta_A}\beta(\bar{\bm{q}}_A\otimes\bm{p}_B) \\
      &+ \frac{1-\eta_B}{\eta_B}\beta(\bm{p}_A\otimes\bar{\bm{q}}_B) 
\end{align}
such that, given $\bar{\bm{q}},\eta_A,\eta_B$, the maximum loophole-free value of the given Bell functional $\beta(\tilde{\bm{p}})$ can be obtained from $\beta_{\QC}(\eta_A,\eta_B)=\max_{\bm{p}\in\QC}\beta_{\eta_A\eta_B}(\bm{p})$ in the following way,
\begin{equation}
    \beta(\tilde{\bm{p}}) = \eta_A\eta_B\beta_\QC(\eta_A,\eta_B)-{(1-\eta_A)(1-\eta_B)}\beta(\bar{\bm{q}}).
\end{equation}
Finally, we define single party Bell functionals, $\beta^{\bar{\bm{q}}}_A(\bm{p}_A),\beta^{\bar{\bm{q}}}_B(\bm{p}_B)$ defined as,
\begin{align} \label{last} \nonumber
    \beta^{\bar{\bm{q}}}_B(\bm{p}_B) &= \beta(\bar{\bm{q}}_A\otimes\bm{p}_B)=\sum_{by}\left(\sum_{ax}\delta_{a,a_x}\beta_{abxy}\right) p_{B}(b|y), \\
    \beta^{\bar{\bm{q}}}_A(\bm{p}_A) &= \beta(\bm{p}_A\otimes\bar{\bm{q}}_B)=\sum_{ax}\left(\sum_{by}\delta_{b,b_y}\beta_{abxy}\right) p_{A}(a|x).
\end{align}
Using \eqref{last} and the function $\alpha(\eta)=(1-\eta)/\eta$ we obtain \eqref{tiltedTheorem}.
\qed

\section{Optimal local assignment strategy for maximum loophole-free violation of the CHSH inequality}\label{appendix_1}
In this section, we demonstrate that the deterministic assignment strategy $\bm{q}\in\LC$ with entries $\bar{q}(ab|xy)=\delta_{a,+1}\delta_{b,+1}$ considered in Section \ref{sec_3} is an optimal choice for all $\eta_A,\eta_B$ to achieve maximal loophole-free nonlocality $C(\tilde{\bm{p}})-2$ in the CHSH scenario. In particular, with the assignment strategy $\bm{q}$, Lemma \ref{tiltedTheorem}, yields the doubly-tilted CHSH inequality \eqref{tilted_CHSH}. However, to find the maximal loophole-free violation of the CHSH inequality $C(\tilde{\bm{p}})-2$, we need to consider all deterministic assignment strategies.

In the CHSH scenario, there are 16 such deterministic assignment strategies, which can be labeled by four bits $s_A, s_B, r_A, r_B\in\{0, 1\}$ specifying the possible deterministic values of the marginals $\mbraket{A_x} = (-1)^{r_A x + s_A}$ and $\mbraket{B_y} = (-1)^{r_B y + s_B}$. For any such a deterministic strategy, using Lemma \ref{tiltedTheorem}, we obtain the following family of doubly-tilted CHSH functionals,
\begin{align} \label{tiltedCHSHgen} \nonumber
C^{(s_A, s_B, r_A, r
_B)}_{\eta_A\eta_B}&(\bm{p}) = \sum_{x, y}(-1)^{x\cdot y}\mbraket{A_x B_y} + \\ \nonumber &
+ \frac{2(1-\eta_B)}{\eta_B}  \sum_x (1+(-1)^{r_B+x})(-1)^{s_B}\mbraket{{A}_x}   \\   & + \frac{2(1-\eta_A)}{\eta_A} \sum_y (1+(-1)^{r_A+y})(-1)^{s_A}\mbraket{B_y}  
\end{align}
Also, for any such deterministic strategy the effective value of the CHSH functional is given by,
\begin{align} \label{effectiveCHSH} \nonumber
    C(\tilde{\bm{p}})=&\eta_A\eta_B C^{(s_A, s_B, r_A, r_B)}_{\eta_A\eta_B}(\bm{p}) \\ &+ 2(1-\eta_A)(1-\eta_B) (-1)^{s_A+s_B+r_Ar_B}.
\end{align}

Consequently, for any specification of the efficiency, $\eta_A,\eta_B$, the optimal deterministic assignment strategy will be the one which yields maximum effective violation \eqref{effectiveCHSH} of the CHSH inequality. 

We can significantly reduce the complexity of finding the optimal deterministic assignment. It is straightforward to verify that those assignments for which $s_A+s_B+r_Ar_B=0$ (mod $2$) yield the same doubly-tilted CHSH inequality \eqref{tilted_CHSH} up to relabeling of inputs and outcomes. The same holds for the other $8$ strategies satisfying $s_A+s_B+r_Ar_B=1$ (mod $2$) which all yield the inequality \eqref{alttilted_CHSH}. Thus, to find the optimal strategy we need only to compare representatives of each class, namely, the cases: $\bm{q}$ such that $q(ab|xy)=\delta_{a,+1}\delta_{b,+1}$ with $r_A=r_B=S_A=S_B=0$ and $\bm{q}'$ such that $q'(ab|xy)=\delta_{a,-1}\delta_{b,+1}$ with $r_A=r_B=S_B=0,s_A=1$. While former strategy $\bm{q}$ yields the doubly-tilted CHSH inequality \eqref{tilted_CHSH} considered in the main text, the later $\bm{q}'$ yields the following distinct doubly-tilted CHSH inequality,
\begin{align}
    C'_{\eta_A\eta_B}(\bm p) &= C(\tilde{\bm p}) + \frac{2}{\eta_B}(1-\eta_B) \mbraket{A_0}-  \frac{2}{\eta_A}(1-\eta_A) \mbraket{B_0}\nonumber\\
    &\leq  2 \left[1- \frac{1}{\eta_A} - \frac{1}{\eta_B} \right]= c'_\LC(\eta_A,\eta_B). \label{alttilted_CHSH}
\end{align}

We use the Jordan's lemma and Gröbner basis based technique described in Section \ref{deriveAnalytic} and Appendix \ref{alpanotbeta}, to retrieve the maximum quantum value $c'_{\QC}(\alpha,\beta)$ of \eqref{alttilted_CHSH} via the following Lemma, with $\alpha=\frac{2}{\eta_B}(1-\eta_B),\beta=\frac{2}{\eta_A}(1-\eta_A)$. 

\begin{lemma}
\label{differentlyTiltedTheorem} [Maximum quantum value of \eqref{alttilted_CHSH}]
    The maximum quantum violation $c'_{\mathcal{Q}}(\alpha,\beta)$ of the doubly-tilted CHSH inequality \eqref{alttilted_CHSH} is the largest root of the degree 6 polynomial,  
\begin{equation}\label{analySolgen1}
    f(\lambda) =  4 \lambda^6 -4 \alpha \beta \lambda^5  + 
\sum_{k=0}^4 \tau'_k \lambda^k,
\end{equation}
where the coefficients $\{\tau'_k\}^{4}_{k=0}$ are polynomials in $\alpha,\beta$, such that, 
\begin{equation}
    \tau'_k =\begin{cases}
        \tau_k,&  \text{ if } k\in\{0,2,4\} \\
        -\tau_k,&  \text{ if } k\in\{1,3,5\},
    \end{cases}
\end{equation}
where $\tau_k$ are polynomials in $\alpha,\beta$ defined in \eqref{alphaNotEqSol}.
\end{lemma}
\begin{proof}
    The proof proceeds exactly in the same way as the proof of Theorem 2.
 \end{proof}
Having access to $c_{\QC}(\eta_A,\eta_B)$ and $c'_{\QC}(\eta_A,\eta_B)$, we plot the difference, 
\begin{align}\label{differenceAssignment}
\Delta_{\eta_A,\eta_B}=&\eta_A\eta_B\left(c_{\QC}(\eta_A,\eta_B)-c'_{\QC}(\eta_A,\eta_B)\right)\\&+4(1-\eta_A)(1-\eta_B),
\end{align} 
between the maximal effective violation of the CHSH inequality \eqref{effectiveCHSH} obtained via the assignment strategies $\bm{q}$ and $\bm{q}'$ against $\eta_A,\eta_B \in(\frac{1}{2},1]$ in FIG. \ref{fig:optAssign}. We find that, for all $\eta_A,\eta_B$ the assignment strategy $\bm{q}$ which yields the doubly-tilted CHSH inequality \eqref{tilted_CHSH} yields higher effective violation of the CHSH inequality, such that $\Delta_{\eta_A,\eta_B}\geq 0$ for all $\eta_A,\eta_B\in(\frac{1}{2},1]$ satisfying \eqref{critical_efficiency}.
\begin{figure} 
    \centering
     \includegraphics[width=\linewidth]{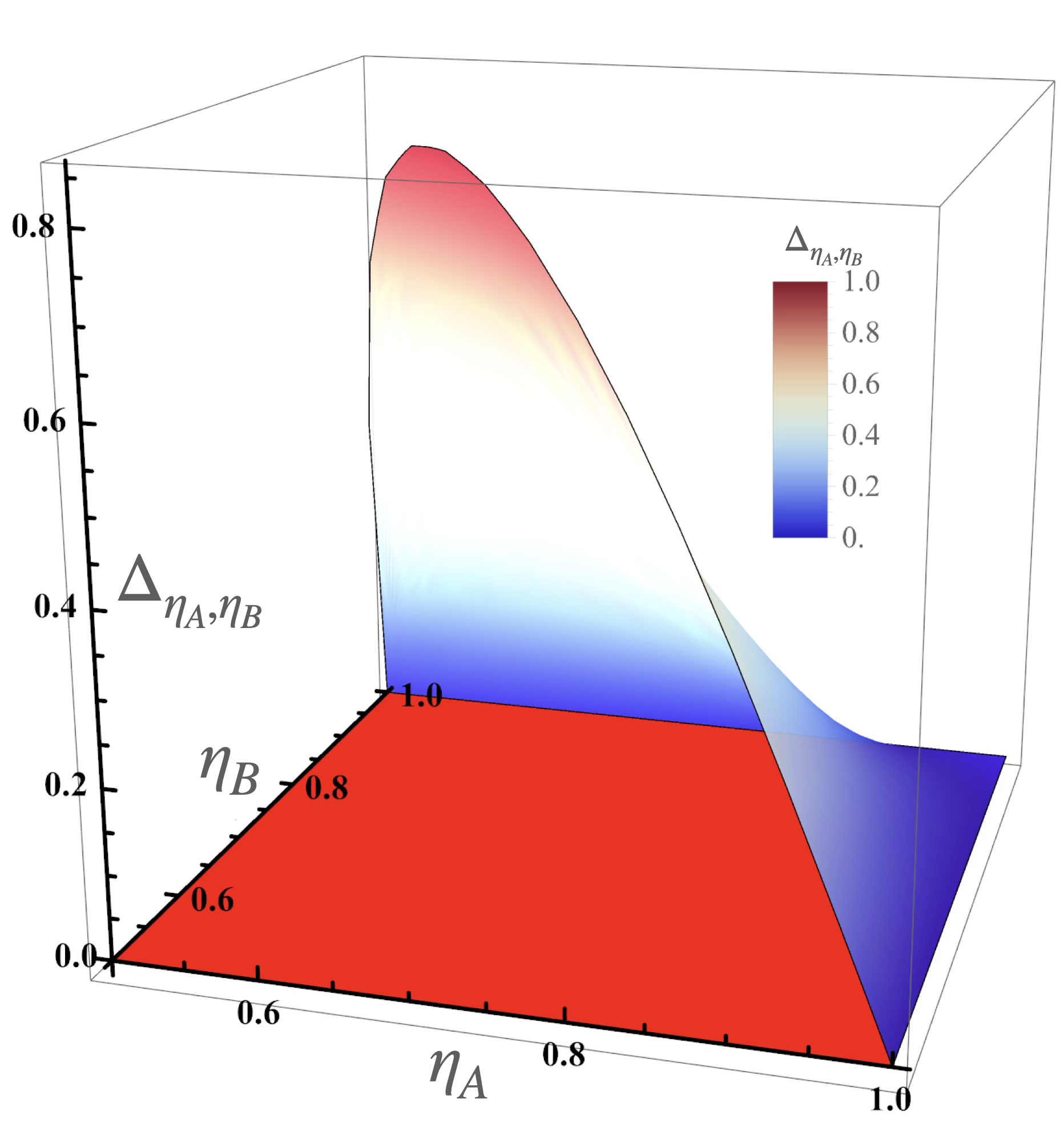}
    \caption{\emph{Optimal assignment strategy for maximum loophole-free violation of the CHSH inequality:---} A plot against efficiencies $\eta_A,\eta_B\in(\frac{1}{2},1]$ satisfying \eqref{critical_efficiency} of the difference $\delta_{\eta_A,\eta_B}$ \eqref{differenceAssignment} in the maximal effective violation of the CHSH inequality \eqref{effectiveCHSH} obtained via the assignment strategy $q(ab|xy)=\delta_{a,+1}\delta_{b,+1}$ and the doubly-tilted CHSH inequality \eqref{tilted_CHSH} and the only other distinct assignment strategy $q'(ab|xy)=\delta_{a,-1}\delta_{b,+1}$ and the distinct doubly-tilted CHSH inequality \eqref{alttilted_CHSH} up to input output relabeling. Since $\Delta_{\eta_A,\eta_B}\geq0$ for $\eta_A,\eta_B\in(\frac{1}{2},1]$ satisfying \eqref{critical_efficiency}, the assignment strategy $\bm{q}$ and the doubly-tilted CHSH inequality \eqref{tilted_CHSH} is the optimal choice for maximal loophole-free nonlocality in the CHSH scenario.}
    \label{fig:optAssign}
\end{figure}

\section{Maximal violation of the general doubly-tilted CHSH inequality and self-testing}\label{alpanotbeta}

In this section, we present and proof the self-testing statements for the maximum violation of doubly-tilted CHSH inequalities \eqref{general_tilted}. In particular, the results presented are the complete generalization (for all $\alpha,\beta\in[0,2]$ such that $0<\alpha+\beta<2$) of the self-testing statements for the symmetrically $(\alpha=\beta)$ tilted CHSH inequalities presented and proved in Section \ref{deriveAnalytic}.

Let us first recall that, since for doubly-tilted CHSH inequalities \eqref{general_tilted}, Alice and Bob have two dichotomic measurements each, Jordan's lemma implies that the optimal local measurements are qubit measurements parametrized as in \eqref{qubitParam} such that the doubly-tilted CHSH operator can be expressed as \eqref{sym_tilted_functional}. We are now prepared to present the main result of this Section, namely, self-testing statements for the maximum quantum violation of doubly-tilted CHSH inequalities, via the following Theorem.

\begin{theorem}\label{doublyTiltedTheorem} [Self-testing with doubly-tilted CHSH inequalities] The maximum quantum violation $c_{\mathcal{Q}}(\alpha,\beta)$ of the doubly-tilted CHSH inequality \eqref{general_tilted}, is the largest root of the degree 6 polynomial,  
\begin{equation}\label{analySolgen}
    f(\lambda) = \sum_{k=0}^6 \tau_k \lambda^k,
\end{equation}
with
\begin{equation} \label{alphaNotEqSol}
    \begin{split}
        \tau_0 =& \alpha^6 \left(27 \beta^2-8\right)+\alpha^4 \left(-54 \beta^4+48 \beta^2 + 32\right) \\
        &+ \alpha^2 \left(27 \beta^6+48 \beta^4-400 \beta^2 + 128\right) \\
        &- 8\beta^6 + 32\beta^4 + 128\beta^2-512. \\
        \tau_1 =& -168\alpha^3\beta^3 + 60 \alpha^5\beta + 160\alpha^3\beta + 60\alpha  \beta^5. \\
        &+ 160\alpha\beta^3-576\alpha\beta. \\
        \tau_2 =& -6\alpha^4\beta^2-6\alpha^2 \beta^4-64\alpha^2\beta^2+20\alpha^4 \\ 
        &+ 96\alpha^2+20\beta^4+96\beta^2 + 320. \\
        \tau_3 =& 8\alpha^3\beta^3-24\alpha^3 \beta -24\alpha\beta^3+96 \alpha\beta. \\
        \tau_4 =& 11 \alpha^2 \beta^2-16 \alpha^2-16 \beta^2-64. \\
        \tau_5 =& 4\alpha\beta. \\
        \tau_6 =& 4.
    \end{split}
\end{equation}
\end{theorem}
Moreover, $C_{\alpha,\beta}(\bm{p})=c_{\mathcal{Q}}(\alpha,\beta)$ \emph{self-tests} a two-qubit quantum strategy with optimal $(*)$ local observables of the form \eqref{qubitParam}, such that Alice's optimal cosine $c^*_A(\alpha,\beta)$ is the following function of the maxmimum quantum value $c_{\QC}(\alpha,\beta)$,
\begin{equation} \label{optimalCAgeneral}
    \begin{split}
    c^*_A(\alpha, \beta)=&\left(32-8\beta^2\right)^{-1} \Bigg[-16-3\alpha^2\beta^2 + 12\alpha^2 + 4\beta^2 \\
    &+ \sqrt{4-\alpha^2}\sqrt{4-\beta^2} \\
    &\sqrt{16-4\alpha^2 4\beta^2 + 9\alpha^2\beta^2 + 16\alpha\beta\, c_\QC(\alpha, \beta)}\Bigg],
    \end{split}
\end{equation}
and Bob's optimal cosine $c^*_B(\alpha,\beta)$ is the following function of Alice's optimal cosine $c^*_A(\alpha,\beta)$, 
\begin{equation} \label{optimalCBgeneral}
    \begin{split}
    c^*_B(\alpha, \beta)=& \frac{\left(4-\beta^2\right) c^*_A(\alpha ,\beta )+\alpha^2-\beta^2}{4-\alpha^2}
    \end{split}
\end{equation}
for all $\alpha,\beta\in[0,2]$ such that $0\leq \alpha+\beta<2$.
\begin{proof}
The proof proceeds on the same lines as the proof of Theorem \ref{selfTestSymmetric} in the main text. The characteristic polynomial of the Bell operator \eqref{sym_tilted_functional} with the parametrization \eqref{qubitParam} has the expression,
\begin{equation} \label{gen_charpoly}
    \begin{split}
        q(\lambda, c_A, c_B, \alpha,& \beta) =\lambda ^4-2 \lambda ^2 \left(\alpha ^2+\beta ^2+4\right)\\
        +&8 \alpha  \beta  \lambda  \left(c_A \left(c_B-1\right)-c_B-1\right)\\
        +&8 \beta ^2 c_A \left(c_B^2-1\right)-8 c_B \left(\alpha ^2+\left(\beta ^2-2\right) c_B\right)\\
        +&\left(\alpha ^2-\beta ^2\right)^2+8 c_A^2 \left(c_B-1\right) \left(\alpha ^2-2 c_B-2\right)
    \end{split}
\end{equation}
where we have rewritten the sines $s_A,s_B$ in terms of the cosines $c_A,c_B$, so that the \eqref{gen_charpoly} is a polynomial function of $c_A,c_B$. For the optimal quantum strategy, $c_A=c^*_A(\alpha,\beta),c_B=c^*_B(\alpha,\beta)$, the maximal quantum value $c_{\QC}(\alpha,\beta)$ corresponds to the maximum eigenvalue $\lambda^*(\alpha,\beta)$ of \eqref{sym_tilted_functional}. Hence, to find $c_{\QC}(\alpha,\beta)$ we need to maximize the largest root of \eqref{gen_charpoly} over the parameters $c_A,c_B\in[0,1)$.

The Lagrangian for this optimization problem is
\begin{equation}
    L = \lambda + s\cdot q(\lambda, c_A, c_B, \alpha, \beta),
\end{equation}
where $s$ is a Lagrange multiplier. A stationary point of $L$ satisfies the conditions $\partial_{\lambda}L = \partial_{c_A}L = \partial_{c_B} L = \partial_s L = 0$, and is therefore a solution of the following polynomial equations,
\begin{subequations}
\begin{gather} \label{gen_poly_conditions}
        \partial_{\lambda}L = 0, \\ \label{gen_poly_conditions2}
        \partial_{c_A}q(\lambda, c_A, c_B, \alpha, \beta) = 0, \\ \label{gen_poly_conditions3}
        \partial_{c_B}q(\lambda, c_A, c_B, \alpha, \beta) = 0, \\ \label{gen_poly_conditions4}
        q(\lambda, c_A, c_B, \alpha, \beta) = 0.
\end{gather}
\end{subequations}
The cosines ${c}_A^*(\alpha,\beta),{c}_B^*(\alpha,\beta)$ parametrizing the optimal measurements, and the maximal quantum value $c_{\QC}(\alpha,\beta)={\lambda}^*(\alpha,\beta)$, are thus common roots of the polynomials $\partial_{\lambda}L,\, \partial_{c_A}q,\, \partial_{c_B}q$, and $q$ in \eqref{gen_charpoly}. The ring defined by these polynomials admits a Gröbner basis that can be computed in Mathematica, which allows the elimination of parameters $c_A$ and $c_B$. This procedure results in a degree $10$ polynomial over $\lambda$. Four of the ten roots of this polynomial do not violate the local bound $2+\alpha+\beta$ for all $\alpha,\beta\in[0,2]$ such that $0\leq\alpha+\beta<2$. Taking the quotient of \eqref{gen_charpoly} with respect to the product of these trivial roots yields $f(\lambda)$ in \eqref{analySolgen}. Yet again, only the largest real root $\lambda^*(\alpha,\beta)$ of $f(\lambda)$ in \eqref{analySolgen} violates the local bound $2+\alpha+\beta$ for all $\alpha,\beta\in[0,2]$ such that $0\leq\alpha+\beta<2$, and hence $c_{\QC}(\alpha,\beta)=\lambda^*(\alpha,\beta)$. 

In particular, in contrast to the symmetric ($\alpha=\beta$) case for which a closed-form solution is presented in the Mathematica notebook \cite{scalaSelfTesting}, the maximal quantum value $c_{\QC}(\alpha,\beta)$ in the general case ($\alpha\neq \beta$) corresponds to the largest real root of a degree 6 polynomial \eqref{analySolgen}, which can be obtained numerically to any desired precision.

To derive the optimal cosines $c^*_{A}(\alpha,\beta),c^*_{B}(\alpha,\beta)$, we consider the polynomial equations \eqref{gen_poly_conditions2}, \eqref{gen_poly_conditions3}. Since, $c_B=1$, $c_A=1$ correspond to compatible measurements, we divide \eqref{gen_poly_conditions2}, \eqref{gen_poly_conditions3} by their respective factors, $8(1-c_B)$, $8(1-c_A)$ to obtain the equations,

\begin{subequations}\label{deriv_polys}
\begin{gather} \label{deriv_polys2}
        (\alpha^2-4c_A)(1+c_B)+2\beta^2 c_A + \alpha\beta\lambda=0, \\ \label{deriv_polys3}
        (\beta^2-4c_B)(1+c_A)+2\alpha^2 c_B + \alpha\beta\lambda=0.  
\end{gather}
\end{subequations}
Taking the difference of \eqref{deriv_polys2} and \eqref{deriv_polys3} to eliminate $\lambda$ we obtain the relation \eqref{optimalCBgeneral}, allowing us to express Bob's optimal cosine $c^*_B(\alpha,\beta)$ as a function of Alice's optimal cosine $c^*_A(\alpha,\beta)$. 

Plugging \eqref{optimalCBgeneral} back into \eqref{deriv_polys2}, we find that Alice's optimal cosine $c^*_A(\alpha,\beta)$ must be a root of the following degree $2$ polynomial,
\begin{equation} \label{finalPolyGen}
    \begin{split}
    r(c_A)=& \alpha \left[\alpha ^2 \beta  \lambda +2 \alpha ^3-\alpha  \left(\beta ^2+4\right)-4 \beta  \lambda\right] \\
    &+ \left(3 \alpha ^2-4\right) \left(\beta ^2-4\right) c_A - 4 \left(\beta
   ^2-4\right) c_A^2.
    \end{split}
\end{equation}
We find that, for all $\alpha,\beta\in[0,2]$ such that $0\leq\alpha+\beta<2$ and $\lambda=\lambda^*(\alpha,\beta)=c_{\QC}(\alpha,\beta)$, Alice's optimal cosine $c^*_A(\alpha,\beta)$ corresponds to largest root of \eqref{finalPolyGen}, given by \eqref{optimalCAgeneral} since for the other root either $c_A\notin [0,1)$ or $c_B\notin [0,1)$. Since, both $c^*_A(\alpha,\beta)$ \eqref{optimalCAgeneral} and $c^*_B(\alpha,\beta)$ (obtained via \eqref{optimalCBgeneral}) are uniquely determined by $c_{\QC}(\alpha,\beta)$, we conclude that the optimal measurements are self-tested by the maximal quantum value $C_{\alpha,\beta}(\bm{p})=c_{\QC}(\alpha,\beta)$. Moreover, we find that the maximum eigenvalue of \eqref{sym_tilted_functional} with the optimal settings $c_A=c^*_A(\alpha,\beta)$ and $c_B=c^*_B(\alpha,\beta)$ is nondegenerate (since no other eigenvalue violates the local bound $2+\alpha+\beta$) for all $\alpha,\beta\in[0,2]$ such that $0\leq\alpha+\beta<2$, which implies that $C_{\alpha,\beta}(\bm{p})=c_{\QC}(\alpha,\beta)$ also self-tests the optimal two-qubit non-maximally entangled state specified by the eigenvector associated with the maximum eigenvalue.  

\end{proof}

\end{document}